%% file: main.tex
\documentclass[12pt,letterpaper]{article}
\usepackage{amssymb,amsmath,amsthm,mathtools}
\usepackage{thmtools,thm-restate}

\usepackage{array,geometry}  
\usepackage{bm,paralist,dsfont,nicefrac,bbm}

\usepackage{etoolbox,algorithm}
\usepackage{algpseudocode}

\usepackage{tikz,float,tabularx}
\usetikzlibrary{arrows}
\usetikzlibrary{arrows.meta}
\usetikzlibrary{shapes}
\usetikzlibrary{backgrounds}
\usetikzlibrary{positioning}
\usetikzlibrary{decorations.markings}
\usetikzlibrary{patterns}
\usetikzlibrary{calc}
\usetikzlibrary{fit}
\usetikzlibrary{math}

\usepackage[driverfallback=hypertex,pagebackref=true,colorlinks]{hyperref}
\hypersetup{linkcolor=[rgb]{.7,0,0}}
\hypersetup{citecolor=[rgb]{0,.7,0}}
\hypersetup{urlcolor=[rgb]{.7,0,.7}}

\usepackage{xcolor} 

\makeatletter
\newcommand{\algmargin}{\the\ALG@thistlm}
\algnewcommand{\parState}[1]{\State%
  \parbox[t]{\dimexpr\linewidth-\algmargin}{\strut #1\strut}}

\newcommand{\alglinenoNew}[1]{\newcounter{ALG@line@#1}}
\newcommand{\alglinenoPop}[1]{\setcounter{ALG@line}{\value{ALG@line@#1}}}
\newcommand{\alglinenoPush}[1]{\setcounter{ALG@line@#1}{\value{ALG@line}}}

\usepackage[capitalize, nameinlink, noabbrev, sort&compress]{cleveref}

\newtheorem{theorem}{Theorem}[section]
\newtheorem*{theorem*}{Theorem}

\newtheorem*{definition*}{Definition}

\newtheorem{lemma}[theorem]{Lemma}
\newtheorem*{lemma*}{Lemma}

\newtheorem{claim}[theorem]{Claim}
\newtheorem*{claim*}{Claim}

\newtheorem{fact}[theorem]{Fact}
\newtheorem*{fact*}{Fact}

\newtheorem*{observation*}{Observation}

\newtheorem*{conjecture*}{Conjecture}

\newtheorem{corollary}[theorem]{Corollary}
\newtheorem*{corollary*}{Corollary}

\newtheorem{remark}[theorem]{Remark}
\newtheorem*{remark*}{Remark}

\newtheorem*{proposition*}{Proposition}

\allowdisplaybreaks

\input{macros.tex}

\title{Settling Nondeterministic Communication Complexity of Graph Connectivity}
\author{
Zhijun Zhang\thanks{\href{mailto:zhijun.zhang@insait.ai}{zhijun.zhang@insait.ai}. This research was partially funded by the Ministry of Education and Science of Bulgaria (support for INSAIT, part of the Bulgarian National Roadmap for Research Infrastructure).}\\
{\small INSAIT, Sofia University ``St. Kliment Ohridski''}
}
\date{} %Remove to get date.

\begin{document}

\maketitle
\thispagestyle{empty}

\begin{abstract}
Determining the communication complexity of graph connectivity has been one of the longest standing open problems in communication complexity since the seminal work of Babai, Frankl, and Simon [FOCS'86].
A simple protocol transmits a spanning forest using $\cO(N\log N)$ communication on $N$-vertex graphs.
While this was shown optimal for deterministic protocols by Hajnal, Maass, and Tur\'an [STOC'88], the optimal nondeterministic and randomized complexities have remained unresolved for decades.

We settle the nondeterministic communication complexity of graph connectivity by proving an optimal $\Omega(N\log N)$ lower bound.
The best prior bound was $\Omega(N\log\log N)$ due to Raz and Spieker [Combinatorica'95].
Our bound is shown by upper bounding the maximum size of $1$-monochromatic rectangles.
Generalizing our techniques to weighted, almost $1$-monochromatic rectangles yields an $\Omega(N \cdot \log(\min(1/\varepsilon,N)))$ lower bound for randomized protocols with error probability $\varepsilon$.
In particular, this bound is optimal for any polynomially small $\varepsilon$.

For constant error probability, we also prove that any $r$-round protocol requires $\Omega(N/r \cdot \log N)$ communication, which is optimal for every constant $r$.
Previously, optimal $\Omega(N\log N)$ lower bounds for graph connectivity were known only for $1$-round [Sun, Woodruff; APPROX/RANDOM'15] and $2$-round protocols [Radhakrishnan, Reddy, Venkat; ITCS'26].
Our proof introduces a new communication problem which can be viewed as a composition of inner product and pointer chasing problems.

Our lower bounds can be extended to other fundamental graph properties, and also imply optimality of known query and streaming algorithms for related problems.
\end{abstract}

\newpage
\pagenumbering{roman} 
\tableofcontents

%\newpage
%\listoffixmes

\newpage
\setcounter{page}{1}
\pagenumbering{arabic}

\input{intro}
\input{overview}
\input{prelim}
\input{nondet}
\input{err}
\input{const}
\input{prob}

\newpage
\section*{AI Methodology}
\label{sec:ai}
\addcontentsline{toc}{section}{\nameref{sec:ai}}

Extensive interaction with ChatGPT was used to assist with exploring, formulating, and improving proof ideas for several lemmas in this paper. 
The tool materially affected \cref{lem:nondet-W,lem:err-bp,lem:prob-D-sublinear}.
In particular,
\begin{enumerate}
    \item ChatGPT discovered the edge coloring argument for proving \cref{lem:nondet-W};
    \item ChatGPT produced a preliminary version of the error-dependent randomized lower bound with worse $\Omega(\log\log(1/\varepsilon))$ dependency on $\varepsilon$.
        The approach was slightly different but to some extent could be viewed as a weaker version of \cref{lem:err-bp} where the additional loss term was exponential in $n$.
        It alone was useful only for exponentially small error probability.
        Direct sum arguments were used to extend it to the full scale of error probabilities.
        The proof has been significantly simplified and improved by human authors to derive the current $\Omega(\log(1/\varepsilon))$ dependency.
    \item ChatGPT suggested the use of \holder{}'s inequality to optimize the lower order term in proving \cref{lem:prob-D-sublinear}.
        It relaxed the round requirement of the lower bound for $\prob_{n,r}$ but did not improve the communication bound for $\prob_{n,r}$ or graph connectivity. 
\end{enumerate}
ChatGPT was also used for proof-reading and minor copy-editing.
Human authors take full responsibility for correctness and exposition of all content in the final manuscript.

\newpage
\bibliographystyle{alpha}
\bibliography{refs}

\newpage
\appendix
\input{info}
\input{pcp}
\input{prob-sublinear}

\end{document}

%% file: macros.tex
\DeclarePairedDelimiter{\ang}{\langle}{\rangle}
\DeclarePairedDelimiter{\paren}{\lparen}{\rparen}

\DeclarePairedDelimiter{\set}{\lbrace}{\rbrace}
\DeclarePairedDelimiter{\card}{\lvert}{\rvert}

\DeclarePairedDelimiter{\abs}{\lvert}{\rvert}
\DeclarePairedDelimiter{\floor}{\lfloor}{\rfloor}
\DeclarePairedDelimiter{\ceil}{\lceil}{\rceil}

\DeclareMathOperator*{\poly}{\mathsf{poly}}
\DeclareMathOperator*{\polylog}{\mathsf{poly\,log}}
\DeclareMathOperator*{\Exp}{\bbE}

\newcommand{\conn}{\mathsf{Conn}}
\newcommand{\bconn}{\mathsf{bConn}}
\newcommand{\ippc}{\mathsf{IP{\circ}PC}}
\newcommand{\prob}{\ippc}
\newcommand{\pc}{\mathsf{PC}}
\newcommand{\bpc}{\mathsf{bPC}}
\newcommand{\pcp}{\mathsf{PC\text{-}Perm}}
\newcommand{\bpcp}{\mathsf{bPC\text{-}Perm}}

\newcommand{\bbD}{\mathbb{D}}
\newcommand{\bbE}{\mathbb{E}}
\newcommand{\bbH}{\mathbb{H}}
\newcommand{\bbI}{\mathbb{I}}

\newcommand{\bfD}{\mathbf{D}}
\newcommand{\bfN}{\mathbf{N}}
\newcommand{\bfR}{\mathbf{R}}

\newcommand{\cA}{\mathcal{A}}
\newcommand{\cB}{\mathcal{B}}
\newcommand{\cC}{\mathcal{C}}
\newcommand{\cE}{\mathcal{E}}
\newcommand{\cO}{\mathcal{O}}
\newcommand{\cP}{\mathcal{P}}
\newcommand{\cS}{\mathcal{S}}
\newcommand{\cU}{\mathcal{U}}
\newcommand{\cV}{\mathcal{V}}
\newcommand{\cW}{\mathcal{W}}

\newcommand{\sA}{\mathsf{A}}
\newcommand{\sB}{\mathsf{B}}
\newcommand{\sL}{\mathsf{L}}
\newcommand{\sR}{\mathsf{R}}
\newcommand{\sX}{\mathsf{X}}
\newcommand{\sY}{\mathsf{Y}}

\newcommand{\oE}{\overline{\cE}}

\newcommand{\wM}{\widehat{M}}
\newcommand{\wa}{\widehat{a}}
\newcommand{\wb}{\widehat{b}}

\newcommand{\wu}{\widehat{u}}
\newcommand{\ww}{\widehat{w}}
\newcommand{\wz}{\widehat{z}}

\newcommand{\1}{\mathds{1}}
\newcommand{\kl}{\bbD_{\mathsf{KL}}}

\newcommand{\bits}{\set{0,1}}
\newcommand{\cyc}{{\sf cyc}}
\newcommand{\supp}{{\sf supp}}

\newcommand{\holder}{H\"older}
\newcommand{\renyi}{R\'enyi}

%% file: intro.tex
\section{Introduction}
\label{sec:intro}

Graph problems have been extensively studied in communication complexity since the early works of \cite{BFS86,HMT88}.
Examples include connectivity~\cite{RS95,FKMSZ08,IKLSW12,SW15,RRV26}, matching~\cite{GKK12,BBEMN22,CJSY26}, vertex/edge coloring~\cite{FM25,CMNS25}, minimum cuts~\cite{ASW25,JNS26}, vertex covers~\cite{DGR25}, planarity~\cite{DP89}, and many more.

One notable problem is connectivity, denoted $\conn_N$.
In the standard two-party communication model~\cite{Yao79}, Alice and Bob receive disjoint edge sets $E_\sA,E_\sB$ of an $N$-vertex undirected graph\footnote{To avoid confusion, we use $N$ to denote the number of vertices in the input graph throughout this paper. $n$ is reserved for the size parameter of auxiliary communication problems used for reduction.} $G=(V,E_\sA \cup E_\sB)$.
The goal is to jointly decide if $G$ is connected.

A simple one-way deterministic protocol solves $\conn_N$ in $\cO(N\log N)$ communication: Alice sends a spanning forest of her edges $E_\sA$ to Bob, who then updates the spanning forest using his edges $E_\sB$ and then decides connectivity of $G$.
\cite{HMT88} proved a matching lower bound, showing that no deterministic protocol with $o(N\log N)$ communication exists, regardless of the number of rounds.
For randomized protocols, the seminal work of \cite{BFS86} gave a reduction from set disjointness to graph connectivity.
When combined with the $\Omega(N)$ lower bound for set disjointness~\cite{KS92,Raz92,BJKS04}, it implies that any randomized protocol for graph connectivity requires $\Omega(N)$ communication.
Closing the gap between lower and upper bounds remains a major open problem for nondeterministic and randomized protocols, which has been highlighted for decades; see for example \cite{BFS86,RS95,IKLSW12,AEGLMN22,RRV26}.

Prior to this work, the strongest nondeterministic lower bound was $\Omega(N\log\log  N)$ due to \cite{RS95}.
Their lower bound actually holds for the more restricted bipartite version $\bconn_N$, where $G$ is promised to be a balanced bipartite graph with $N/2$ vertices on each side, and both $E_\sA$ and $E_\sB$ are guaranteed to be perfect matchings across the bipartition\footnote{Technically, $\bconn_N$ may induce multigraphs with parallel edges. This can be easily fixed by splitting every vertex into a length-$1$ path, resulting in simple bipartite graphs with $\Theta(N)$ vertices.}.
Equivalently, $E_\sA,E_\sB$ can be viewed as permutations $a,b$ over $[N/2]$, and $G$ is connected if and only if $b^{-1}a$ is a single $(N/2)$-cycle.

No superlinear $\omega(N)$ lower bound is known for randomized protocols in general.
When the number of rounds is bounded, the first progress was an optimal $\Omega(N\log N)$ lower bound for $1$-round protocols due to \cite{SW15}.
Recently, the same lower bound was shown for $2$-round protocols by \cite{RRV26}.
For every constant $r \ge 3$, \cite{RRV26} also gave an $\Omega(N\log^{(2r+1)}N)$ lower bound for $r$-round protocols.

\subsection{Our Results}

Our first main contribution is settling the nondeterministic communication complexity of graph connectivity.
\cite{RS95} explicitly conjectured that $\Omega(N\log N)$ communication is required for nondeterministic protocols.
We resolve the conjecture in the affirmative.

\begin{theorem}[Informal version of \cref{thm:nondet}]
\label{thm:nondet-informal}
The nondeterministic communication complexity of $N$-vertex graph connectivity is $\Theta(N\log N)$.
\end{theorem}

Like many nondeterministic communication lower bounds, the proof of \cref{thm:nondet-informal} is via upper bounds on the maximum size of $1$-monochromatic rectangles.
In \cref{sec:overview-nondet}, we sketch our new ideas and give a brief comparison with the approach of \cite{RS95}.

For randomized protocols, transcripts do not necessarily correspond to monochromatic rectangles anymore.
Nevertheless, our techniques can be generalized to show that $\bconn_N$ does not have large weighted rectangles that are close to being $1$-monochromatic.
This yields an error-dependent randomized lower bound.

\begin{theorem}[Informal version of \cref{thm:err}]
\label{thm:err-informal}
The randomized communication complexity of $N$-vertex graph connectivity with error probability $\varepsilon=\varepsilon(N)$ is $\Omega(N \cdot \log(\min(1/\varepsilon,N)))$.
\end{theorem}

It is worth noting that \cref{thm:err-informal} implies an optimal $\Omega(N\log N)$ lower bound for any polynomially small error probability $\varepsilon=1/\poly(N)$.
For any $\varepsilon=o(1)$, it gives a superlinear $\omega(N)$ lower bound, improving the $\Omega(N)$ bound of \cite{BFS86}.
We also remark that any improvement to the $\Omega(\log(1/\varepsilon))$ dependency on $\varepsilon$ would immediately imply $\omega(N)$ lower bounds even for protocols with constant error.
This is because $\cO(\log(1/\varepsilon))$ repetitions suffice to decrease the error probability from constant to $\varepsilon$.

For randomized protocols with constant error, we also prove tradeoffs between the number of rounds and the total communication.

\begin{theorem}[Informal version of \cref{thm:const}]
\label{thm:const-informal}
The $r$-round randomized communication complexity of $N$-vertex graph connectivity with constant error probability is $\Omega(N/r \cdot \log N)$.
\end{theorem}

In particular, for every constant $r$, \cref{thm:const-informal} implies optimal $\Omega(N\log N)$ lower bounds for any $r$-round protocols, improving \cite{SW15,RRV26}.
Also note that our lower bound is always superlinear in the case of $r=o(\log N)$, while an $\Omega(N/r \cdot \log N)$ lower bound for $r=\Omega(\log N)$ already follows from the $\Omega(N)$ bound of \cite{BFS86}, which holds even for protocols with unbounded rounds.

We emphasize that the proof of \cref{thm:const-informal} is not a generalization of previous approaches.
Indeed, the lower bounds in \cite{RRV26} crucially rely on a specific variant of pointer chasing.
\cref{sec:pcp} presents a round-efficient protocol that can jump over quadratically many layers for their specific variant, which tends to be an unfavorable feature for lower bound proofs\footnote{With current techniques, it is unclear how to prove substantially stronger lower bounds for these problems when the number of rounds is significantly smaller than the number of layers.}.
\cref{sec:overview-const} includes a more detailed discussion.

Motivated by this, our proof of \cref{thm:const-informal} introduces a new two-party communication problem $\prob_{n,r}$, which is essentially a composition of two well-studied communication problems, namely inner product and pointer chasing, as the name suggests.
Concretely, the input of $\prob_{n,r}$ consists of two bit vectors $X,Y \in \bits^n$ and $2r$ permutations $A_1,\dots,A_r,B_1,\dots,B_r$ over $[n]$.
Alice holds all $A_i,B_i$ for odd $i$, while Bob holds $X,Y$ as well as all $A_i,B_i$ for even $i$.
The goal is to output the binary inner product $\ang{A_r \cdots A_1 X,B_r \cdots B_1 Y}$, where $A_1,\dots,A_r,B_1,\dots,B_r$ are identified with their corresponding permutation matrices.
At a high level, before outputting the final inner product, protocols have to simultaneously ``chase'' each of the $2n$ input bits through $r$ layers of permutations (i.e., ``pointers'').

We prove the following strong lower bound\footnote{The communication lower bound in \cref{thm:prob-informal} is independent of the number of rounds, so long as $r$ is not too large; see \cref{thm:prob,thm:prob-sublinear} for the precise range of $r$.} for $\prob_{n,r}$.
Unlike most of the known lower bounds for pointer chasing problems (see \cref{sec:related}), \cref{thm:prob-informal} neither loses any $\poly(r,\log n)$ factor nor has iterated-logarithm-type dependency on $r$.
We believe that $\prob_{n,r}$ can be a powerful alternative to pointer chasing and may be of independent interest.

\begin{theorem}[Informal version of \cref{thm:prob,thm:prob-sublinear}]
\label{thm:prob-informal}
The $r$-round randomized communication complexity of $\prob_{n,r}$ with constant error probability is $\Omega(n\log n)$, for $r \ll n$.
\end{theorem}

We remark that our lower bounds for graph connectivity can be extended to other fundamental graph properties.
In \cref{thm:property}, we give lower bounds similar to \cref{thm:const-informal} for deciding if the input graph is bipartite, cycle-free, Eulerian, Hamiltonian, or planar.

Our lower bounds also immediately imply optimality of known query and streaming algorithms.
In particular, \cref{thm:nondet-informal} means that the zero-error cut query algorithm for connectivity in \cite[Theorem 1.2]{AEGLMN22} has optimal expected query complexity $\Theta(N)$.
Their algorithm actually outputs a spanning forest, whereas our lower bound holds even for algorithms that only decide connectivity; see \cref{cor:query} for details.

Another important corollary is that any constant-pass streaming algorithm for global minimum cuts in unweighted graphs, also known as edge connectivity, requires $\Omega(N\log N)$ space.
It implies that the two-pass algorithm in \cite{AD21} is indeed optimal, even in logarithmic factors.
This is formalized in \cref{cor:stream}.

\subsection{Related Work}
\label{sec:related}

\paragraph{Pointer chasing.}

Initially introduced by \cite{PS84}, the problem of pointer chasing $\pc_{n,k}$ has played a central role in communication complexity.
Specifically, the problem instance consists of $k$ functions $f_1,\dots,f_k: [n] \to [n]$, together with some fixed starting vertex $v_0$.
Alice owns $f_i$ for odd $i$ and Bob owns $f_i$ for even $i$.
For $i \in [k]$, recursively define the $i$-th target vertex $v_i=f_i(v_{i-1})$.
The goal is for both parties to output $v_k$.
For the bit version $\bpc_{n,k}$, protocols only need to decide the parity of $v_k$.

In the original version of \cite{PS84}, Alice holds a single function $f_\sA$, and $f_1 = f_3 = \dots = f_\sA$.
Similarly for Bob, $f_2 = f_4 = \dots = f_\sB$.
A straightforward protocol uses only $\cO(k\log n)$ communication in $k$ rounds, by exchanging target vertices $v_i$.
A long line of work has thus focused on protocols with $r \le k-1$ rounds, or $k$ rounds but starting with Bob.
\cite{NW93} gave a $(k-1)$-round randomized protocol using $\cO((n/k+k) \cdot \log n)$ communication.
On the lower bound side, a series of results \cite{NW93,Kla00,Yeh20,MYZ25} leads to an $\Omega(n/k+k)$ randomized lower bound even for the bit version.
Additionally, for constant $k$, \cite{DJS98} showed a $k$-round deterministic protocol starting with Bob that uses $\cO(n\log^{(k-1)} n)$ communication, which is shown to be tight in \cite{PRV01} via a matching $\Omega(n\log^{(k-1)} n)$ randomized lower bound.

For reductions to connectivity and other related graph problems, it is also useful to study the variants of pointer chasing where each function $f_i$ is promised to be a permutation; see for example \cite{AN21,RRV26}.
We denote these variants by $\pcp_{n,k},\bpcp_{n,k}$.

\paragraph{Beyond constant error.}

The standard randomized communication complexity concerns protocols with some fixed constant error probability, typically $1/3$.
In many cases, however, it is also important to understand tradeoffs between communication and error probability.
That is, we want to obtain error-dependent communication bounds.
Examples\footnote{The first two cases are also closely related to direct product theorems and XOR lemmas.} include error probabilities close to $1$ (i.e., small success probabilities)~\cite{BRWY13,BW15,JPY16}, close to $1/2$ (i.e., small advantages)~\cite{BM13,Yu22}, close to $0$~\cite{BGPW13,HPZZ21,LMW23}, and asymmetric error probabilities~\cite{GW16}.

\paragraph{Other computation models.}

Communication complexity has proven to be an extremely useful tool in a wide variety of computation models.
For example, via standard connections, strong results are known for streaming~\cite{FKMSZ08,GKK12,SW15,GO16,AD21,AKZ24,AKNS24,ABKNS26,EKSZ26}, sketching~\cite{NY19,AKO20,Yu21,AKZ25}, and query algorithms~\cite{AEGLMN22,BBEMN22,ASW25,JNS26}.

As a fundamental graph property, connectivity in particular has been studied in many different settings other than two-party communication.
The most closely related one is multi-party communication~\cite{PVZ16,WZ17}.
Besides, efficient algorithms and strong lower bounds for graph connectivity have also been established for streaming~\cite{FKMSZ08,SW15}, sketching~\cite{AGM12,NY19,Yu21}, and query algorithms~\cite{ACK21,AEGLMN22}.

\paragraph{Quantum communication complexity.}

With the rapid development of quantum computation, quantum communication complexity has attracted increasing attention as well.
Strong quantum bounds are shown for many classical communication problems, such as graph connectivity~\cite{IKLSW12,EKNP14}, pointer chasing~\cite{JRS02a,JRS02b,KNTZ07}, inner product~\cite{Kre95,NS06,CDNT13}, and set disjointness~\cite{BCW98,HW02,Raz03,JRS03,AA05,KSW07,BGKMT18}.

%% file: overview.tex
\section{Technical Overview}
\label{sec:overview}

We provide a high-level overview of our proofs in this section.
\cref{sec:overview-nondet} is dedicated to the nondeterministic lower bound and compares our approach with \cite{RS95}.
Towards the end, we also briefly discuss how to generalize our techniques to prove error-dependent lower bounds for randomized protocols.
\cref{sec:overview-const} first motivates our new two-party communication problem $\prob$ by a comparison with the pointer-chasing-based approach of \cite{RRV26}, and then sketches the intuition for the communication lower bound on $\prob$.

\subsection{Nondeterministic Lower Bound}
\label{sec:overview-nondet}

Nondeterministic communication complexity is by definition closely related to monochromatic rectangles.
In particular, upper bounds on the maximum size of $1$-monochromatic rectangles imply nondeterministic communication lower bounds.

Let $n=N/2$.
The approach of \cite{RS95} focused on a special collection $\cS$ of $n!=2^{n\log n-\cO(n)}$ permutation pairs $(a,b) \in \bconn_N^{-1}(1)$ such that any $1$-monochromatic rectangle $\cA \times \cB$ contains at most $(n/\polylog(n))^n=2^{n\log n-\Omega(n\log\log n)}$ elements of $\cS$.
This yields their $\Omega(n\log\log n)$ nondeterministic lower bound.
A closer look reveals that all permutation pairs $(a,b) \in \cS$ actually have distinct $a$ (and distinct $b$).
Thus, the argument of \cite{RS95} upper bounds the projection of $(\cA \times \cB) \cap \cS$ onto either coordinate.
If we additionally assume that every row (column) of $\cA \times \cB$ contains some element of $\cS$, \cite{RS95} essentially bounds the maximum height (width) of certain $1$-monochromatic rectangles, i.e., \[\max_{\cA \times \cB \subseteq \bconn_N^{-1}(1)} \card{\cA} = \max_{\cA \times \cB \subseteq \bconn_N^{-1}(1)} \card{\cB},\] subject to the assumption on $\cA \times \cB$.
Instead, we directly upper bound the (logarithmic) size of all $1$-monochromatic rectangles: \[S_n = \max_{\cA \times \cB \subseteq \bconn_N^{-1}(1)} \log\card{\cA \times \cB}.\]

Consider two independent random variables $A \sim \cA$ and $B \sim \cB$, where $\cA \times \cB$ is the maximizer in the above definition, so $\bbH(A,B)=S_n$.
In the bipartite graph view, $A,B$ induce a Hamiltonian cycle $C$ of length $N=2n$ among left vertices $[n]_\sL$ and right vertices $[n]_\sR$.
Moreover, the edges of $C$ are provided by $A,B$ alternately.

The starting point of our proof is the following observation (\cref{lem:nondet-recur}).
Suppose we know one edge $(u_\sL,v_\sR)$ supplied by $A$, so $A(u)=v$.
Its two adjacent edges\footnote{These are all distinct edges as $B(u) \ne v$. Otherwise, the graph is not Hamiltonian for $n>1$.} $(u_\sL,B(u)_\sR)$ and $(B^{-1}(v)_\sL,v_\sR)$ must be supplied by $B$.
Imagine ``shortcutting'' these three edges by a new edge $(B^{-1}(v)_\sL,B(u)_\sR)$.
Then, we actually get another cycle of length $2(n-1)=N-2$ on all vertices except $u_\sL,v_\sR$.
In other words, up to renaming of the remaining vertices, we 
have reduced the case to $1$-monochromatic rectangles for a smaller input size.

Let $A_{(u,v)}: [n] \setminus \set{u} \to [n] \setminus \set{v}$ denote the restriction of $A$ on $[n] \setminus \set{u}$.
Also let $B_{(u,v)}: [n] \setminus \set{u} \to [n] \setminus \set{v}$ be the restriction of $B$ on $[n] \setminus \set{u}$ except redefining $B_{(u,v)}(B^{-1}(v))=B(u)$.
Note that specifying $A(u)=v$ reveals at most $\log n$ bits of information on average while $\bbH(A_{(u,v)},B_{(u,v)})$ is at most $S_{n-1}$ by the argument above.
As $A$ is fully determined given $A(u)=v$ and $A_{(u,v)}$, it now boils down to upper bounding the remaining randomness in $B$.

Equivalently, we want to show that there cannot be too many $b' \ne b \in \cB$ satisfying $b'_{(u,v)}=b_{(u,v)}$ for any fixed $b \in \cB$.
A trivial upper bound of $n-1$ immediately follows from the observation that $b_{(u,v)}$ specifies $n-1$ edges for $b'$, out of which only one is newly introduced by  ``shortcutting''.
However, an increase of $\log n + \log(n-1)$ for each recursive step leads to an upper bound of $\log(n!)+\log((n-1)!)$ on $S_n$.
Meanwhile, it can be verified that $\card{\bconn_N^{-1}(1)}=n!(n-1)!$.
No meaningful result can be derived from such a bound.

The key to a much tighter bound is realizing that $\cA \times \cB$ being $1$-monochromatic actually imposes more restrictions on $b'$, as $\bconn_N(a,b')=1$ must also hold for all $a \in \cA$.
A careful analysis (\cref{lem:nondet-bp}) establishes two necessary conditions for $b'_{(u,v)}=b_{(u,v)}$:
\begin{enumerate}
    \item It must be that\footnote{See \cref{sec:prelim} for the cycle decomposition representation of permutations and also \cref{fig:nondet-bp} for a detailed illustration of the first condition.} $b^{-1}b'=(u\ w\ b^{-1}(v))$ for some $w \in [n] \setminus \set{u,b^{-1}(v)}$;
    \item Moreover, $u,b^{-1}(v),w$ must follow the same cyclic order $(u\ b^{-1}(v)\ w)$ along the cycle $b^{-1}a$ for every $a \in \cA$.
\end{enumerate}
The first condition essentially rephrases the requirement in the trivial upper bound above.
Intuitively, the more important second condition is useful because the cyclic order of many triples of vertices must be fixed if we want to have a lot of $w$ satisfying the condition.
On the other hand, a large $\cA$ means high randomness in the cycle $b^{-1}a$, which in turn implies that most triples of vertices should not have a fixed cyclic order.
The conflict between these two requirements suggests a better upper bound.

Formalizing this intuition involves a delicate edge coloring argument (\cref{lem:nondet-W}).
After all, we manage to show that the remaining randomness in $B$ conditioned on $B_{(u,v)}$ is in fact bounded away from $\log n$ by at least\footnote{Here, we actually assume $\card{\cA} \ge \card{\cB}$. This is without loss of generality by symmetry.} $\Omega(\bbH(A)/n)=\Omega(S_n/n)$, up to lower order terms, when averaged over $(u,v)$.
Putting everything together, we have that \[S_n \lesssim \log n + S_{n-1} + \log n - \Omega\paren*{\frac{S_n}{n}}.\]
Solving the recurrence then yields \[S_n \lesssim (2-\Omega(1)) \cdot n\log n,\] concluding our proof.

As to randomized protocols, we generalize our techniques to upper bound the size of any weighted, almost $1$-monochromatic rectangle.
Specifically, suppose \[\Pr(\bconn_N(A,B)=0) \le \delta,\] for independent random permutations $A,B$.
We prove an upper bound (\cref{lem:err-S}) on $\bbH(A,B)$, similar to the nondeterministic one, with an additional $\cO(\sqrt{\delta} \cdot \poly(n))$ loss.
Indeed, a typical $A_{(u,v)} \times B_{(u,v)}$ is also $\delta$-close to $1$-monochromatic on average.
Thus, the recursive approach for the nondeterministic setting still applies.
The only source of the additional $\cO(\sqrt{\delta} \cdot \poly(n))$ loss comes from upper bounding the remaining randomness in $B$ (\cref{lem:err-bp}).
Overall, as typical transcripts on input from $\bconn_N^{-1}(1)$ correspond to weighted, almost $1$-monochromatic rectangles, optimal $\Omega(n\log n)$ lower bounds can be derived for polynomially small error probabilities.
The full lower bound for the entire scale of error probabilities then follows by repetition.

\subsection{Randomized Lower Bound with Constant Error}
\label{sec:overview-const}

The $2$-round lower bound of \cite{RRV26} for randomized protocols solving connectivity crucially relies on the bit version $\bpcp_{n,k}$ of pointer chasing with permutations.
Ordinarily, the literature around pointer chasing primarily concerns the separation between $(k-1)$-round and $k$-round protocols.
Different variants are typically irrelevant to the communication complexity, up to logarithmic factors, in such a case.
Upper bounds mentioned in \cref{sec:related} all apply to $\pc_{n,k}$, while many lower bounds can be adapted to hold even for $\bpcp_{n,k}$.

However, the situation changes dramatically if we do care about the logarithmic gap.
\cite{RRV26} gave a protocol for solving\footnote{Their protocols are stated for $\pcp_{n,k}$, but also work for $\pc_{n,k}$ in general.} $\pc_{n,k}$ with $\cO(nk\log\log n)$ communication in roughly $k/3$ rounds.
The round complexity can be reduced to about $k/4$ for the bit version.
No protocols with significantly fewer rounds are known.
In sharp contrast, we present in \cref{sec:pcp} round-efficient protocols for solving $\pcp_{n,k}$ and $\bpcp_{n,k}$ with $\cO(nk\log\log n)$ communication in only $\cO(\sqrt{k})$ rounds.

The $2$-round lower bound of \cite{RRV26} for connectivity follows from an optimal $\Omega(n\log n)$ communication lower bound for $2$-round\footnote{To better align with our convention for graph connectivity protocols, we assume throughout this paper that only the receiver of the last message has to output. The results in \cite{RRV26} were stated in the standard context of pointer chasing, where both parties need to output. This changes the number of rounds by $1$. Both their protocol and our protocol start with Bob.} protocols solving $\bpcp_{n,7}$.
Both their protocol and our protocol can solve $\bpcp_{n,6}$ with $\cO(n\log\log n)$ communication in $2$ rounds.
With $3$ rounds and $\cO(n\log\log n)$ communication, their protocol can only solve $\bpcp_{n,10}$, whereas our protocol can solve $\bpcp_{n,12}$.
The smallest hard instance that has not been ruled out is thus $\bpcp_{n,13}$. 
For more rounds and general $k$, to the best of our knowledge, it is unknown if $\Omega(n\log n)$ communication lower bounds can be proved for any variant of pointer chasing, even for protocols with far fewer than $\sqrt{k}$ rounds.

The protocol in \cref{sec:pcp} essentially preprocesses disjoint intervals of layers in parallel.
For each interval, almost balanced partitions into $\Theta(\log n)$ parts are constructed for each layer such that they are consistent with the permutations, at the communication cost of $\cO(n\log\log n)$ per layer.
Then, revealing the starting vertex allows the protocol to jump over each interval while only requiring the exact pointers for vertices within the same part.
In a nutshell, the protocol is only forced to send full information about a $\Theta(1/\log n)$ fraction of the pointers in each layer, once the starting vertex is fixed.

This motivates our new two-party communication problem $\prob_{n,r}$.
Recall that protocols for $\prob_{n,r}$ have to output $\ang{A_r \cdots A_1 X, B_r \cdots B_1 Y}$, where $X,Y$ are bit vectors and $A_1,\dots,A_r,B_1,\dots,B_r$ are permutations.
Intuitively, permutations $A_1,\dots,A_r$ are expected to send a random $X$ to a random bit vector for the final inner product.
Similarly for $B_1,\dots,B_r$ and $Y$.
Therefore, if some protocol only chased a small number of coordinates, it should not be able to predict the inner product very well.

The proof of the communication lower bound for $\prob_{n,r}$ involves delicate information-theoretic arguments.
The high level plan is straightforward: we keep track of the joint distribution of $A_i \cdots A_1 X$ and $B_i \cdots B_1 Y$ in round $i$ during execution of the protocol, by bounding its deviation from uniform.
Nevertheless, proving the tight bound, in particular ensuring no $\poly(r,\log n)$ factor loss, turns out to be quite nontrivial and requires extensive use of \renyi{} entropies $\bbH_\alpha(\cdot)$ (specifically, for $\alpha=2,\infty$, namely collision entropy and min-entropy).
Plenty of further intuitions and more technical details are provided throughout the full proof in \cref{sec:prob}.

Once lower bounds for $\prob$ are shown, our reduction to graph connectivity uses a layered construction similar to those in \cite{AN21,RRV26}.

%% file: prelim.tex
\section{Preliminaries}
\label{sec:prelim}

\paragraph{Notation.}

For $n \ge 1$, we write $[n]$ as a shorthand for $\set{1,\dots,n}$.
For any tuple $x=(x_1,x_2,\dots)$ and $i \ge 1$, define $x_{<i}$ to be $(x_1,\dots,x_{i-1})$.
Similarly define $x_{\le i}$, $x_{\ge i}$, etc.
For two bit vectors $x,y$, we write $\ang{x,y}$ for their binary inner product.
$\circ$ denotes composition of functions.
Logarithms are always base 2.
Iterated logarithm of order $k$ is denoted $\log^{(k)}(\cdot)$.

\paragraph{Permutations.}

For $n \ge 1$, an $n$-permutation $\sigma$ is a bijection from $[n]$ to $[n]$.
Let $\cP_n$ be the collection of all $n$-permutations, so $\card{\cP_n}=n!$.
We also use the standard cycle decomposition representation for permutations.
For example, $\sigma=(1\ 4)(2\ 10\ 5)$ means that $\sigma(1)=4$, $\sigma(4)=1$, $\sigma(2)=10$, etc., and all other elements (of omitted $1$-cycles) map to themselves.
For uniqueness, the smallest element of each cycle is written first, and the cycles are arranged in increasing order of their smallest elements.
The number of cycles in the decomposition of $\sigma$ is denoted by $\cyc(\sigma)$.
Let $\cC_n \subseteq \cP_n$ be the collection of all $n$-cycles, which are exactly $n$-permutations $\sigma$ with $\cyc(\sigma)=1$, so $\card{\cC_n}=(n-1)!$.
Additionally, with slight abuse of notation, an $n$-permutation $\sigma$ can be equivalently represented by an $n \times n$ binary matrix where $\sigma_{j,i}=\1(j=\sigma(i))$ (i.e., mapping the $i$-th standard basis vector to the $\sigma(i)$-th).
So $\sigma^{-1}=\sigma^\top$.
The specific representation we use will be clear from context.

\paragraph{Information theory.}

For any distribution $Q$, the support of $Q$ is denoted by $\supp(Q)$.
We use $x \sim Q$ to mean an outcome drawn according to $Q$.
For any set $S$, we also write $x \sim S$ when $x$ is drawn from $S$ uniformly at random.
With slight abuse of notation, we may write a random variable $X$ interchangeably with its distribution.
For $\alpha \in [0,\infty]$, the \renyi{} entropy of order $\alpha$ of a random variable $X$ is denoted by $\bbH_\alpha(X)$.
The cases of $\alpha=1,2,\infty$ are also known as Shannon entropy, collision entropy, and min-entropy, respectively.
We use $\bbD_\alpha(X \Vert Y)$ to denote the \renyi{} divergence of order $\alpha$ of $X$ from $Y$.
The case of $\alpha=1$ is also known as Kullback-Leibler (KL) divergence.
When $Y$ follows the uniform distribution over a universe $\cU \supseteq \supp(X)$, we simply write $\bbD_{\alpha}(X)=\bbD_{\alpha}(X \Vert Y) = \log\card{\cU} - \bbH_\alpha(X)$.
When there is no confusion, we use the convention $\bbH(\cdot)=\bbH_1(\cdot)$ and $\bbD(\cdot \Vert \cdot)=\kl(\cdot \Vert \cdot)=\bbD_1(\cdot \Vert \cdot)$.
The mutual information between two random variables $X$ and $Y$ is denoted by $\bbI(X;Y)$.
\cref{sec:info} provides a summary of the necessary background and tools we use in information theory.
Throughout, random variables are denoted by uppercase letters, and we use the corresponding lowercase letters for their outcomes.

\paragraph{Communication complexity.}

We use the standard two-party communication model in this paper.
Deterministic communication complexity is denoted by $\bfD(\cdot)$ while nondeterministic communication complexity is denoted by $\bfN(\cdot)$.
For $r \ge 1$ and $\varepsilon \in [0,1]$, we use $\bfR^{(r)}_\varepsilon(\cdot)$ for the $r$-round randomized communication complexity with error probability $\varepsilon$.
When the number of rounds is unbounded, we simply write $\bfR_\varepsilon(\cdot)$.
For any constant error probability $\varepsilon \in (0,1/2)$, the subscript is omitted, as it can always be assumed without loss of generality to be $1/3$ by repetition.
We use the convention that Alice sends to Bob in the first round and that output is only required for the receiver of the last message, unless explicitly stated otherwise.
Throughout, communication is always measured in the worst case.

%% file: nondet.tex
\section{Nondeterministic Lower Bound for Connectivity}
\label{sec:nondet}

In this section, we prove \cref{thm:nondet}.
The upper bound follows from a simple one-way deterministic protocol that communicates to Bob a spanning forest of Alice's edges.

\begin{theorem}
\label{thm:nondet}
For even $N>1$, the nondeterministic communication complexity of $\bconn_N$ is $\bfN(\bconn_N) = \Theta(N \log N)$, regardless of the number of rounds.
\end{theorem}

As an immediate corollary, we obtain an optimal query lower bound for graph connectivity in the zero-error setting.
It implies that the cut query algorithm in \cite[Theorem 1.2]{AEGLMN22} is optimal, even if only connectivity has to be decided (without outputting a spanning forest).

\begin{corollary}
\label{cor:query}
For $N \ge 1$, any zero-error randomized cut query algorithm that decides connectivity in $N$-vertex graphs, requires at least $\Omega(N)$ queries in expectation.
\end{corollary}

\begin{proof}
By Markov's inequality, any zero-error protocol with $o(N)$ queries in expectation can be transformed into another one with $o(N)$ queries in the worst case, that has one-sided constant error only in the connected case.
As the total number of possible inputs is no more than $2^{\cO(N^2)}$, the use of randomness can be reduced to $\cO(\log N)$ by Newman's theorem~\cite{New91}, while preserving one-sided constant error.
In the meantime, each cut query can be simulated using $\cO(\log N)$ communication.
Altogether, it implies a randomized protocol with $o(N\log N)$ communication and one-sided constant error, contradicting \cref{thm:nondet}, since nondeterministic protocols are at least as easy.
\end{proof}

We prove \cref{thm:nondet} in the rest of this section.
Fix sufficiently large $n=N/2$.
The total number of input combinations $a,b \in \cP_n$ with $\bconn_N(a,b)=1$ is $T=n!(n-1)!$.
To see this, pick an arbitrary permutation $b \in \cP_n$ and an arbitrary cycle $c \in \cC_n$.
Then, $\bconn_N(bc,b)=1$ as $b^{-1}(bc)=c$.
Furthermore, $(bc,b)$ enumerates exactly all possibilities in $\bconn_N^{-1}(1)$, and there are $\card{\cP_n}=n!$ choices for $b$ and $\card{\cC_n}=(n-1)!$ choices for $c$.

Let $\cA \times \cB \subseteq \cP_n \times \cP_n$ be any $1$-monochromatic rectangle for $\bconn_N$.
We want to show that the size of $\cA \times \cB$ must be substantially smaller than $T$.
In particular, any bound of the form $\log\card{\cA \times \cB} \le (2-\Omega(1)) \cdot n\log n$ would suffice to prove \cref{thm:nondet}.
This is because $\log T = (2-o(1)) \cdot n\log n$, implying that at least $2^{\Omega(n\log n)}$ $1$-monochromatic rectangles are required to cover $\bconn_N^{-1}(1)$.

To this end, define \[S_n = \max_{\cA \times \cB \subseteq \bconn_N^{-1}(1)} \log\card{\cA \times \cB}.\]
\cref{lem:nondet-S} shows a recurrence for $S_n$ that is strong enough to obtain the desired bound.

\begin{lemma}
\label{lem:nondet-S}
For $n>1$, it holds that \[\paren*{1+\frac{1}{4n}} \cdot S_n \le S_{n-1} + 2\log n + \cO(\log\log n).\]
\end{lemma}

We first prove \cref{thm:nondet}, assuming \cref{lem:nondet-S}, whose proof is deferred to \cref{sec:nondet-S}.

\begin{proof}[Proof of \cref{thm:nondet}]
For some sufficiently large constant $K$, we solve the recurrence in \cref{lem:nondet-S} as follows:
\begin{align}
S_n
& \le \frac{4n}{4n+1} \cdot S_{n-1} + \frac{4n}{4n+1} \cdot \paren*{2\log n + K\log\log n} \nonumber\\
& \le \frac{4n(4n-4)}{(4n+1)(4n-3)} \cdot S_{n-2} + \paren*{\frac{4n}{4n+1}+\frac{4n(4n-4)}{(4n+1)(4n-3)}} \cdot \paren*{2\log n + K\log\log n} \nonumber\\
& \qquad \vdots \nonumber\\
& \le \paren*{\prod_{i \in [2,n]} \frac{4i}{4i+1}} \cdot S_1 + \paren*{\sum_{i \in [2,n]} \prod_{j \in [i,n]} \frac{4j}{4j+1}} \cdot \paren*{2\log n + K\log\log n} \nonumber\\
& \le \paren*{\sum_{i \in [n]} \prod_{j \in [i,n]} \frac{4j}{4j+1}} \cdot \paren*{2\log n + K\log\log n}, \label{eqn:nondet-S}
\end{align}
as $S_1=\log 1=0$.
Define \[f_n = \sum_{i \in [n]} \prod_{j \in [i,n]} \frac{4j}{4j+1}.\]
Then, we have that $f_1=4/5$ and that $f_n = 4n/(4n+1) \cdot (f_{n-1}+1)$.
By induction on $n$, it can be verified that $f_n=4n/5$.
Plugging it back into \cref{eqn:nondet-S} yields \[S_n \le \frac{8}{5} n\log n + \cO(n\log\log n).\]

Observe that at least $T/2^{S_n}$ $1$-monochromatic rectangles are required to cover $\bconn_N^{-1}(1)$.
Consequently, the nondeterministic communication complexity of $\bconn_N$ is at least \[\log\frac{T}{2^{S_n}} = \log T - S_n \ge (2-o(1)) \cdot n\log n - \frac{8}{5} n\log n - \cO(n\log\log n) = \Omega(N \log N),\] by Stirling's approximation, as claimed.
\end{proof}

\subsection{Proof of \cref{lem:nondet-S}}
\label{sec:nondet-S}

We now prove the recurrence in \cref{lem:nondet-S}.
Fix $\cA \times \cB$ to be the maximizer in the definition of $S_n$.
Without loss of generality, assume $\card{\cA} \ge \card{\cB}$.
Consider two independent and uniformly random variables $A \sim \cA$ and $B \sim \cB$, so we have that \[\bbH(A)+\bbH(B)=\bbH(A,B)=S_n \qquad \text{and} \qquad \bbH(A) \ge \frac{S_n}{2}.\]
Fix $u,v \in [n]$.
Suppose $a(u)=v$ for some $a \in \cA$.
Then, $b(u) \ne v$ must hold for all $b \in \cB$, as otherwise $b^{-1}a$ contains a singleton cycle, contradicting $(a,b) \in \cA \times \cB \subseteq \bconn_{N}^{-1}(1)$.
Define $a_{(u,v)}: [n] \setminus \set{u} \to [n] \setminus \set{v}$ to be the restriction of $a$ on $[n] \setminus \set{u}$.
Also define $b_{(u,v)}: [n] \setminus \set{u} \to [n] \setminus \set{v}$ to be the same as $b$ except that $b_{(u,v)}(b^{-1}(v)) = b(u)$.
That is, $b_{(u,v)}$ redirects the preimage of $v$, namely $b^{-1}(v) \ne u$, to the image of $u$, namely $b(u) \ne v$.
For any $u,v \in [n]$ such that $a(u)=v$ for some $a \in \cA$, define \[\cA_{(u,v)} = \set{a_{(u,v)} \mid a \in \cA, a(u)=v}, \qquad \text{and} \qquad \cB_{(u,v)} = \set{b_{(u,v)} \mid b \in \cB}.\]

\cref{lem:nondet-recur} shows that $\cA_{(u,v)} \times \cB_{(u,v)}$ also constitutes a $1$-monochromatic rectangle for $n-1$, after renaming both $[n] \setminus \set{u}$ and $[n] \setminus \set{v}$ to $[n-1]$.

\begin{lemma}
\label{lem:nondet-recur}
For any $u,v \in [n]$ such that $a(u)=v$ for some $a \in \cA$, it must be that $\wb^{-1}\wa$ consists of a single $(n-1)$-cycle on $[n] \setminus \set{u}$, for all $(\wa,\wb) \in \cA_{(u,v)} \times \cB_{(u,v)}$.
\end{lemma}

\begin{proof}
Fix any $a \in \cA$ such that $a(u)=v$, and $b \in \cB$.
Consider a bipartite graph with vertex set $[n]$ on each side.
We use subscripts $\sL,\sR$ to distinguish vertices on two sides.
Note that each of $a$ and $b$ corresponds to a perfect matching, and their union forms a single $2n$-cycle, due to $(a,b) \in \bconn_N^{-1}(1)$.
Furthermore, $b(u)_\sR,u_\sL,v_\sR,b^{-1}(v)_\sL$ actually constitute a length-$3$ subpath on the cycle.
The net effect of $a_{(u,v)},b_{(u,v)}$ is simply replacing this subpath with an edge between $b(u)_\sR$ and $b^{-1}(v)_\sL$.
As a result, the edges corresponding to $a_{(u,v)},b_{(u,v)}$ form a single $2(n-1)$-cycle on vertices excluding $u_\sL,v_\sR$, concluding the proof.
\end{proof}

After conditioning on $A(u)=v$, observe that $A_{(u,v)}$ is now a uniformly random element of $\cA_{(u,v)}$.
However, $B_{(u,v)}$ is not necessarily uniform over $\cB_{(u,v)}$ as each $b_{(u,v)}$ may arise from a different number of $b \in \cB$.
Nevertheless, \cref{lem:nondet-recur} implies that \[\log\card{\cA_{(u,v)} \times \cB_{(u,v)}} \le S_{n-1}.\]
As will be seen later, this is exactly the source of the $S_{n-1}$ term in the recurrence in \cref{lem:nondet-S}.
In addition, since $A(u)$ has at most $n$ possibilities, its uncertainty contributes one $\log n$ term.
What remains is to bound the contribution coming from multiplicity of $b_{(u,v)}$.

To this end, for $x \ne y \in [n]$, let $\cW_{(x,y)}(b)$ be the collection of $w \in [n] \setminus \set{x,y}$ such that $(w\ x\ y)$ is compatible with the cycle $b^{-1}a$ for every $a \in \cA$.
That is, $w,x,y$ follow the cyclic order $(w\ x\ y)$, instead of the only other possibility $(w\ y\ x)$, along $b^{-1}a$.

\cref{lem:nondet-bp} upper bounds the multiplicity of $b_{(u,v)}$.

\begin{lemma}
\label{lem:nondet-bp}
For any $u,v \in [n]$ such that $a(u)=v$ for some $a \in \cA$, and $b \in \cB$, it holds that \[\card{\set{b' \in \cB \mid b'_{(u,v)}=b_{(u,v)}}} \le \card{\cW_{(u,b^{-1}(v))}(b)}+1.\]
\end{lemma}

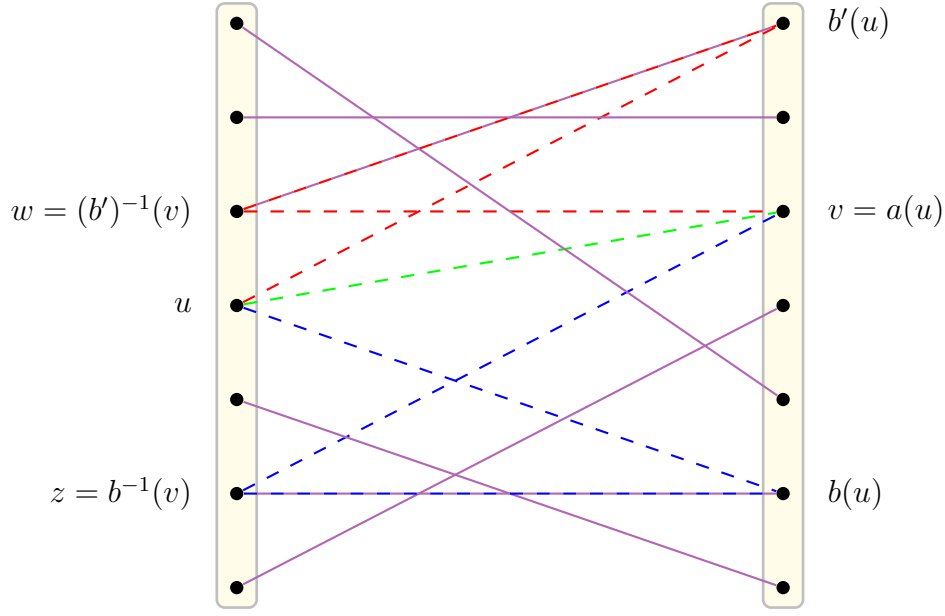
\begin{figure}[tp]
    \centering
    \input{fig/nondet-bp}
    \caption{An illustration of the $6$-cycle in the proof of \cref{lem:nondet-bp}, with $n=7$. Red and blue edges correspond to $b'$ and $b$, respectively. The green edge corresponds to $a$. Solid purple (=red+blue) edges correspond to $\wb=b_{(u,v)}=b'_{(u,v)}$. Dashed edges are not included in $\wb$.}
    \label{fig:nondet-bp}
\end{figure}

\begin{proof}
Let $\wb=b_{(u,v)}$ and $z=b^{-1}(v) \ne u$ for notational convenience.
We use the same bipartite graph view as in the proof of \cref{lem:nondet-recur}.
The net effect from $b^{-1}a$ to $b_{(u,v)}^{-1}a_{(u,v)}$ is replacing the length-$3$ subpath $(b(u)_\sR,u_\sL,v_\sR,z_\sL)$ by the edge between $b(u)_\sR$ and $z_\sL$, so $\wb(z)=b(u)$.
Similarly, for any $b' \in \cB$ such that $b'_{(u,v)}=b_{(u,v)}=\wb$, the net effect from $(b')^{-1}a$ to $(b'_{(u,v)})^{-1}a_{(u,v)}$ is replacing the length-$3$ subpath $(b'(u)_\sR,u_\sL,v_\sR,w_\sL)$ by the edge between $b'(u)_\sR$ and $w_\sL$, where $w=(b')^{-1}(v) \ne u$, so $\wb(w)=b'(u)$.

Observe that $b,b',\wb$ all share the same image on $[n] \setminus \set{u,w,z}$.
If $w=z$, then we additionally have $b'(u)=\wb(w)=\wb(z)=b(u)$ and $b'(w)=v=b(z)$ by definition.
As a result, it must be that $b'=b$.

Now suppose $w \ne z$.
$b$ provides edges $(z_\sL,v_\sR)$ and $(u_\sL,b(u)_\sR)$, while $b'$ provides edges $(w_\sL,v_\sR)$ and $(u_\sL,b'(u)_\sR)$.
In the meantime, due to $b_{(u,v)}=\wb=b'_{(u,v)}$, we know that $b'(z)=b'_{(u,v)}(z)=\wb(z)=b(u)$, so $b'$ also provides the edge $(z_\sL,b(u)_\sR)$.
Similarly, we have that $b(w)=b_{(u,v)}(w)=\wb(w)=b'(u)$, so $b$ also provides the edge $(w_\sL,b'(u)_\sR)$.
Altogether, we obtain a $6$-cycle $(u_\sL,b'(u)_\sR,w_\sL,v_\sR,z_\sL,b(u)_\sR)$ using alternating edges provided by $b',b$.
As all other edges corresponding to $\wb$ are common to $b,b'$, it actually implies that $b^{-1}b'=(u\ w\ z)$.
See \cref{fig:nondet-bp} for an illustration of the $6$-cycle.

Furthermore, not all $w \in [n] \setminus \set{u,z}$ are valid.
In particular, it must be that \[(b')^{-1}a = (u\ w\ z)^{-1}b^{-1}a = (u\ z\ w)b^{-1}a \in \cC_n,\] for every $a \in \cA$, due to $(a,b') \in \bconn_N^{-1}(1)$.
Suppose $u,w,z$ follow the cyclic order $(u\ w\ z)$ along $b^{-1}a$.
Equivalently, as $b^{-1}a \in \cC_n$, we have that $b^{-1}a = (u\ \wu\ w\ \ww\ z\ \wz)$ for some potentially empty sequences $\wu,\ww,\wz$ such that $\wu,\ww,\wz$ constitute a disjoint partition of $[n] \setminus \set{u,w,z}$.
Then, we can get \[(b')^{-1}a = (u\ \wu)(w\ \ww)(z\ \wz),\] consisting of three disjoint cycles, a contradiction.

On the other hand, if $u,w,z$ follow the cyclic order $(u\ z\ w)$ along $b^{-1}a$ instead, then we have that $b^{-1}a = (u\ \wu\ z\ \wz\ w\ \ww)$.
It implies that \[(b')^{-1}a = (u\ \wu\ w\ \ww\ z\ \wz)\] still consists of a single cycle.
Since the above analysis holds for every $a \in \cA$, the maximum number of such $w$, or equivalently $b'$, is exactly $\card{\cW_{(u,z)}(b)}$ by definition.

Together with the possibility of $b'=b$, this concludes the proof.
\end{proof}

Intuitively, if the cycle $b^{-1}a$ is close to random, $\cW_{(u,b^{-1}(v))}(b)$ should be relatively small, as otherwise the cyclic order of many triples $(w,u,b^{-1}(v))$ is already fixed.
\cref{lem:nondet-W} formalizes this intuition, providing the last major piece to the proof of \cref{lem:nondet-S}.

\begin{lemma}
\label{lem:nondet-W}
For any $b \in \cB$, it holds that \[\Exp_{u \sim [n]} \Exp_{a \sim A} \log \paren*{\card{\cW_{(u,b^{-1}(a(u)))}(b)}+1} \le \log n - \frac{\bbH(A)}{2n} + \cO(\log\log n).\]
\end{lemma}

\begin{proof}
Let $C=b^{-1}A$.
Note that $C \in \cC_n$ because of $(A,b) \in \bconn_N^{-1}(1)$.
For $u \in [n]$, we say that vertex $u$ has an outgoing edge to $C(u)$ and an incoming edge from $C^{-1}(u)$.
In the following, we will color all $n(n-1)$ potential directed edges such that the color $\phi_u$ of $u$'s outgoing edge and the color $\psi_u$ of $u$'s incoming edge, for all vertices $u$, uniquely determine the cycle $C$.

Partition all potential edges $(u,v)$ into $g=\cO(\log n)$ groups such that the $i$-th group consists of all edges $(u,v)$ satisfying $2^{i-1} \le \card{\cW_{(u,v)}(b)}+1 < 2^i$.
This is well-defined because $0 \le \card{\cW_{(u,v)}(b)} \le n-2$ by definition.
For each group $i$, independently sample $w^{(i)}_1,\dots,w^{(i)}_{t_i} \sim [n]$, where $t_i=\Theta(n/2^i \cdot \log n)$.
For any edge $(u,v)$ in group $i$, by Chernoff bound, there exists $j \in [t_i]$ such that $w^{(i)}_j \in \cW_{(u,v)}(b) \cup \set{u}$, with probability at least $1-1/\poly(n)$.
Define the color of edge $(u,v)$ to be $(i,j)$.
By union bound over all $n(n-1)$ potential edges, such $w^{(i)}_1,\dots,w^{(i)}_{t_i}$ for every group $i$ and thus the overall coloring can be fixed.

Consider $U \sim [n]$ and $V = C(U)$.
Let $(I,J)$ denote the color of $(U,V)$.
It holds that
\begin{align*}
\bbH(I,J)
& = \bbH(I) + \bbH(J \mid I) \tag{by chain rule of entropy (\cref{fact:chain-H})}\\
& \le \log g + \sum_{i \in [g]} \Pr(I=i) \cdot \log t_i \\
& \le \cO(\log\log n) + \sum_{u,v \in [n]} \Pr((U,V)=(u,v)) \cdot \log \frac{\cO(n\log n)}{\card{\cW_{(u,v)}(b)}+1} \\
& \le \log n + \cO(\log\log n) - \Exp_{u \in [n]} \Exp_{c \sim C} \log \paren*{\card{\cW_{(u,c(u))}(b)}+1} \\
& = \log n + \cO(\log\log n) - \Exp_{u \in [n]} \Exp_{a \sim A} \log \paren*{\card{\cW_{(u,b^{-1}(a(u)))}(b)}+1}.
\end{align*}
Note that $(U,V)$ is a uniformly random edge of $C$.
Thus, for outgoing edges, we can get
\begin{align*}
\Exp_{u \sim [n]} \bbH(\phi_u)
& = \Exp_{u \sim [n]} \bbH(I,J \mid U=u) \\
& = \bbH(I,J \mid U) \\
& \le \bbH(I,J) \\
& \le \log n + \cO(\log\log n) - \Exp_{u \in [n]} \Exp_{a \sim A} \log \paren*{\card{\cW_{(u,b^{-1}(a(u)))}(b)}+1}.
\end{align*}
A similar argument applies to incoming edges.

It now remains to prove that $C$ is indeed uniquely determined by $\phi_u,\psi_u$ for all $u \in [n]$.
For every color $(i,j)$, we show that all edges in $C$ with color $(i,j)$ are determined by \[\cU_{(i,j)} = \set{u \in [n] \mid \phi_u=(i,j)} \qquad \text{and} \qquad \cV_{(i,j)} = \set{v \in [n] \mid \psi_v=(i,j)}.\]
That is, we need to recover a perfect matching between $\cU$ and $\cV$ such that every matching edge $(u,v)$ has color $(i,j)$ and satisfies $v=C(u)$.
For any $u \in \cU$ and $v \in \cV$, the edge $(u,v)$ having color $(i,j)$ means $w^{(i)}_j \in \cW_{(u,v)}(b) \cup \set{u}$ by construction of the coloring.
By definition of $\cW_{(u,v)}(b)$, it further implies that either $w^{(i)}_j=u$, or $(w^{(i)}_j\ u\ v)$ is compatible with every $b^{-1}a$ with $a \in \cA$, including $C$.
In other words, if we break the cyclic order of $C$ right before $w^{(i)}_j$ into a linear order, $v$ must come after $u$.

The perfect matching between $\cU,\cV$ induced by $C$ satisfies that $v=C(u)$ comes immediately after $u$ in the linear order, so the difference in their rankings is exactly $1$.
Observe that the sum of these differences over all matching edges is fixed, regardless of the perfect matching.
As a result, any other perfect matching would have some matching edge $(u,v)$ with negative difference in their rankings, meaning that $v$ comes before $u$, a contradiction to the color of edge $(u,v)$ being $(i,j)$.
Consequently, we can construct a bipartite graph between $U$ and $V$, and add an edge $(u,v)$ if and only if it has color $(i,j)$.
The unique perfect matching then recovers all edges in $C$ with color $(i,j)$, as claimed.

This concludes the proof by rearranging the terms, as \[\bbH(A) = \bbH(C) \le \sum_{u \in [n]} \bbH(\phi_u) + \sum_{u \in [n]} \bbH(\psi_u).\]
\end{proof}

Now, we are ready to prove \cref{lem:nondet-S}.

\begin{proof}[Proof of \cref{lem:nondet-S}]
For any $u \in [n]$, $(A,B)$ can be equivalently represented as \[(A(u),A_{(u,A(u))},B_{(u,A(u))},\beta_{(u,A(u))}(B)),\] where $\beta_{(u,A(u))}(B) \in [\card{\cW_{(u,B^{-1}(A(u)))}(B)}+1]$ is a unique label among all $b' \in \cB$ such that $b'_{(u,A(u))}=B_{(u,A(u))}$.
This is well-defined because of \cref{lem:nondet-bp}.
Now observe that
\begin{align*}
\bbH(A,B)
& = \bbH\paren*{A(u),A_{(u,A(u))},B_{(u,A(u))},\beta_{(u,A(u))}(B)} \\
& = \bbH\paren*{A(u)} + \bbH\paren*{A_{(u,A(u))},B_{(u,A(u))} \mid A(u)}+ \bbH\paren*{\beta_{(u,A(u))}(B) \mid A,B_{(u,A(u))}} \tag{by chain rule of entropy (\cref{fact:chain-H})}\\
& \le \log n + S_{n-1} + \bbH(\beta_{(u,A(u))}(B) \mid A,B_{(u,A(u))}) \tag{by \cref{lem:nondet-recur} and averaging over $v=A(u)$}\\
& \le \log n + S_{n-1} + \Exp_{(a,b) \sim (A,B)} \log \paren*{\card{\cW_{(u,b^{-1}(a(u)))}(b)}+1}.
\end{align*}
By \cref{lem:nondet-W} and averaging over $b \sim \cB$ and $u \sim [n]$, it implies that \[\bbH(A,B) \le \log n + S_{n-1} + \log n - \frac{\bbH(A)}{2n} + \cO(\log\log n).\]
The proof is completed by rearranging the terms, since $\bbH(A) \ge S_n/2$ by assumption.
\end{proof}

%% file: fig/nondet-bp.tex
\begin{tikzpicture}

\tikzset{
    layer/.style={rectangle,rounded corners=3pt,draw,gray!50,line width=1pt,fill=yellow!10,inner sep=4pt},
    vertex/.style={circle,fill,minimum size=5pt,inner sep=0pt}
}

\tikzmath{
    \n=7;
    \u=4;\v=3;\w=3;\x=1;\z=6;\y=6;
    \m{1}=5;\m{2}=2;\m{3}=1;\m{5}=7;\m{6}=6;\m{7}=4;
}

\begin{scope}[local bounding box=bl]
    \node[vertex] (l1) {};
    \foreach \i in {2,...,\n} {
        \pgfmathtruncatemacro{\ip}{\i-1}
        \node[vertex] (l\i) [below=30pt of l\ip] {};
    }
\end{scope}

\begin{scope}[local bounding box=br]
    \node[vertex] (r1) [right=200pt of l1] {};
    \foreach \i in {2,...,\n} {
        \pgfmathtruncatemacro{\ip}{\i-1}
        \node[vertex] (r\i) [below=30pt of r\ip] {};
    }
\end{scope}

\begin{scope}[on background layer]
    \draw[layer] ($(bl.south west) - (5pt,5pt)$) rectangle ($(bl.north east) + (5pt,5pt)$);
    \draw[layer] ($(br.south west) - (5pt,5pt)$) rectangle ($(br.north east) + (5pt,5pt)$);
\end{scope}

\draw[thick,green,dashed,dash pattern=on 5pt off 5pt] (l\u) -- (r\v);

\draw[thick,red,dashed,dash pattern=on 5pt off 5pt] (l\w) -- (r\v);
\draw[thick,red,dashed,dash pattern=on 5pt off 5pt] (l\u) -- (r\x);

\draw[thick,blue,dashed,dash pattern=on 5pt off 5pt] (l\z) -- (r\v);
\draw[thick,blue,dashed,dash pattern=on 5pt off 5pt] (l\u) -- (r\y);

\foreach \i in {1,2,5,7} {
    \draw[thick,violet!60] (l\i) -- (r\m{\i});
}

\draw[thick,red,dashed,dash pattern=on 5pt off 5pt] (l\w) -- (r\x);
\draw[thick,violet!60,dashed,dash pattern=on 5pt off 5pt, dash phase=5pt] (l\w) -- (r\x);

\draw[thick,blue,dashed,dash pattern=on 5pt off 5pt] (l\z) -- (r\y);
\draw[thick,violet!60,dashed,dash pattern=on 5pt off 5pt, dash phase=5pt] (l\z) -- (r\y);

\node [left=10pt of l\u] {$u$};
\node [left=10pt of l\w] {$w=(b')^{-1}(v)$};
\node [left=10pt of l\z] {$z=b^{-1}(v)$};

\node [right=10pt of r\v] {$v=a(u)$};
\node [right=10pt of r\x] {$b'(u)$};
\node [right=10pt of r\y] {$b(u)$};

\end{tikzpicture}

%% file: err.tex
\section{Lower Bound for Connectivity with Small Error}
\label{sec:err}

This section is dedicated to the proof of error-dependent randomized lower bounds for connectivity, as shown in \cref{thm:err}.

\begin{theorem}
\label{thm:err}
For even $N>1$, and $\varepsilon=\varepsilon(N) \in (0,1/3)$, any randomized protocol for $\bconn_N$ with error probability $\varepsilon$ requires at least $\bfR_\varepsilon(\bconn_N) = \Omega(N \cdot \log(\min(1/\varepsilon,N)))$ communication, regardless of the number of rounds.
\end{theorem}

Fix sufficiently large $n=N/2$, and $\varepsilon=\varepsilon(N) \in (0,1/3)$.
The proof of \cref{thm:err} can be viewed as a generalization of the analysis in \cref{sec:nondet}.
In a nutshell, the key idea is to show that $\bconn_N$ does not even admit large almost $1$-monochromatic rectangles, including weighted ones.
Concretely, the distance of a product distribution\footnote{That is, $Q$ is exactly the product of its two marginals.} $Q$ on $\cP_n \times \cP_n$ from being $1$-monochromatic is defined to be \[\Delta(Q) = \Pr_{(a,b) \sim Q}\paren*{\bconn_N(a,b)=0}.\]
For $\delta \in [0,1]$, define \[S_{n,\delta} = \max_{Q: \Delta(Q) \le \delta} \bbH(Q),\]
where $Q$ ranges over product distributions only.

\cref{lem:err-S} is the analog of \cref{lem:nondet-S}.
In fact, the special case of $\delta=0$ recovers the nondeterministic setting.

\begin{lemma}
\label{lem:err-S}
For $n \ge 1$ and $\delta \in [0,1]$, it holds that \[S_{n,\delta} \le S'_{n,\delta},\] where $S'_{n,\delta}$ is defined recursively by
\[\paren*{1+\frac{1}{4n}} \cdot S'_{n,\delta} = S'_{n-1,\delta} + 2\log n + K \cdot \paren*{\log\log n + \sqrt{\delta}n^3\log n},\] for $n > 1$, starting with $S'_{1,\delta}=0$, for some sufficiently large constant $K$.
\end{lemma}

Intuitively, on input $(A,B) \in \bconn_N^{-1}(1)$, the protocol should mostly produce transcripts that lead to output $1$ and that correspond to almost $1$-monochromatic rectangles.
We now prove \cref{thm:err}, assuming \cref{lem:err-S}.
We will prove \cref{lem:err-S} in \cref{sec:err-S}.

\begin{proof}[Proof of \cref{thm:err}]
With $S_{1,\delta}=S'_{1,\delta}=0$, a similar argument to the analysis in the proof of \cref{thm:nondet} solves the recurrence in \cref{lem:err-S} to
\begin{align}
S_{n,\delta} \le S'_{n,\delta} \le \frac{8}{5} n\log n + \cO\paren*{n\log\log n + \sqrt{\delta}n^4\log n}. \label{eqn:err-S}
\end{align}

First suppose for the sake of contradiction that there exists a randomized protocol with error probability $\varepsilon \le n^{-10}$ and communication $L=o(n\log n)$.
We will later see that the case of $\varepsilon > n^{-10}$ follows by repetition.
It suffices to consider independent and uniformly random input $A,B \sim \cP_n$.
We emphasize, however, that we do not derandomize the protocol and insist that it has error probability $\varepsilon$ on every possible input.

Let $\Gamma=(M,R)$ be the transcript of the protocol, where $M,R$ denote the message and public randomness, respectively.
Since $A$ and $B$ are still independent conditioned on $\Gamma$, with slight abuse of notation, we write $\Delta(\Gamma)=\Delta(A,B \mid \Gamma)$.
We also assume the output $O=O(\Gamma)$ is a function of $\Gamma$, as otherwise the receiver of the last message can send its output using one extra bit of communication.
Replacing $L$ with $L+1$ does not change its asymptotic behavior.
In addition, $\Gamma$ is said to be typical if $\Delta(\Gamma)$ is no more than $\delta=n^{-8}$.
Let $T=T(\Gamma)$ denote the indicator of $\Gamma$ being typical.
For notational convenience, also let $Z=\bconn_N(A,B)$, so $\Pr(Z=0 \mid \gamma) \le \delta$ for any $\gamma$ such that $T(\gamma)=1$.

Note that $(A,B)$ conditioned on $Z=1$ is just the uniform distribution\footnote{But it is not a product distribution.} on $\bconn_N^{-1}(1)$, which has size $n!(n-1)!$ as shown in \cref{sec:nondet}.
Therefore, we can get
\begin{align}
\bbH(A,B \mid Z=1) = 2n\log n - \cO(n), \label{eqn:err-H}
\end{align}
by Stirling's approximation.
In the meantime, by definition, observe that
\begin{align*}
\Exp_{\gamma \mid Z=1} O \cdot \Delta(\gamma)
& = \sum_{(a,b) \in \bconn_N^{-1}(1)} \Pr\paren*{(A,B)=(a,b) \mid Z=1} \cdot \Exp_{\gamma \mid (a,b)} O \cdot \Delta(\gamma) \\
& = n \cdot \sum_{(a,b) \in \bconn_N^{-1}(1)} \Pr\paren*{(A,B)=(a,b)} \cdot \Exp_{\gamma \mid (a,b)} O \cdot \Delta(\gamma) \\
& \le n \cdot \sum_{(a,b) \in \bconn_N^{-1}(1)} \Pr\paren*{(A,B)=(a,b)} \cdot \Exp_{\gamma \mid (a,b)} \Pr_{(a',b') \mid \gamma} \paren*{\bconn_N(a',b') \ne O} \\
& \le n \cdot \Exp_{a,b \sim \cP_n} \Exp_{\gamma \mid (a,b)} \Pr_{(a',b') \mid \gamma} \paren*{\bconn_N(a',b') \ne O} \\
& = n \cdot \Exp_\gamma \Exp_{(a,b) \mid \gamma} \Pr_{(a',b') \mid \gamma} \paren*{\bconn_N(a',b') \ne O} \\
& = n \cdot \Exp_\gamma \Pr_{(a',b') \mid \gamma} \paren*{\bconn_N(a',b') \ne O} \\
& = n \cdot \Pr_{a',b' \sim \cP_n} \Exp_{\gamma \mid (a',b')} \paren*{\bconn_N(a',b') \ne O} \\
& \le n \varepsilon.
\end{align*}
By Markov's inequality, it implies that \[\Pr(O=1,T=0 \mid Z=1) = \Pr_{\gamma  \mid Z=1}(O \cdot \Delta(\gamma) > \delta) \le \frac{\Exp_{\gamma \mid Z=1} O \cdot \Delta(\gamma)}{\delta} \le \frac{n\varepsilon}{\delta}.\]
Since both $O$ and $T$ are functions of $\Gamma$ only, $A$ and $B$ remain independent conditioned on any fixed $\gamma$ such that $O(\gamma)=1$ and $T(\gamma)=1$.
By averaging over all such $\gamma$, \cref{eqn:err-S} implies that \[\bbH(A,B \mid \Gamma,O=1,T=1) \le \frac{8}{5} n\log n + \cO(n\log\log n),\] as $\delta=n^{-8}$.
Furthermore, it must be that
\begin{align*}
& \bbH(A,B \mid \Gamma,O=1,T=1) \\
& \hspace{1cm} = \bbH(A,B,Z \mid \Gamma,O=1,T=1) \\
& \hspace{1cm} = \bbH(A,B \mid \Gamma,Z,O=1,T=1) + \bbH(Z \mid \Gamma,O=1,T=1) \tag{by chain rule of entropy (\cref{fact:chain-H})}\\
& \hspace{1cm} \ge \Pr(Z=1 \mid O=1,T=1) \cdot \bbH(A,B \mid \Gamma,Z=1,O=1,T=1) \\
& \hspace{1cm} \ge (1-\delta) \cdot \bbH(A,B \mid \Gamma,Z=1,O=1,T=1),
\end{align*}
implying that
\begin{align*}
\bbH(A,B \mid \Gamma,Z=1,O=1,T=1)
& \le \frac{1}{1-\delta} \cdot \bbH(A,B \mid \Gamma,O=1,T=1) \\
& \le \frac{8}{5} n\log n + \cO(n\log\log n),
\end{align*}
due to $\delta=n^{-8}$ as well.
Altogether, we can get
\begin{align}
& \bbH(A,B \mid \Gamma,Z=1) \nonumber\\
& \hspace{1cm}\le \bbH(A,B,O,T \mid \Gamma,Z=1) \nonumber\\
& \hspace{1cm} = \bbH(O,T \mid \Gamma,Z=1) + \bbH(A,B \mid \Gamma,O,T,Z=1) \tag{by chain rule of entropy (\cref{fact:chain-H})} \nonumber\\
& \hspace{1cm}\le 2 + \bbH(A,B \mid \Gamma,O=1,T=1,Z=1) \nonumber\\
& \hspace{2cm} + \Pr\paren*{(O,T) \ne (1,1) \mid Z=1} \cdot \bbH(A,B \mid \Gamma,(O,T) \ne (1,1),Z=1) \nonumber\\
& \hspace{1cm}\le \frac{8}{5} n\log n + \cO(n\log\log n) \nonumber\\
& \hspace{2cm} + \paren*{\Pr(O=0 \mid Z=1) + \Pr(O=1,T=0 \mid Z=1)} \cdot \bbH(A,B \mid Z=1) \nonumber\\
& \hspace{1cm}\le \frac{8}{5} n\log n + \cO(n\log\log n) + \paren*{\varepsilon + \frac{n\varepsilon}{\delta}} \cdot \cO(n\log n) \tag{by \cref{eqn:err-H}} \nonumber\\
& \hspace{1cm}\le \frac{8}{5} n\log n + \cO(n\log\log n). \label{eqn:err-H-gamma}
\end{align}
As the total length of all messages is at most $L$, combining \cref{eqn:err-H,eqn:err-H-gamma} yields
\begin{align*}
L
& \ge \bbH(M \mid R,Z=1) \\
& \ge \bbI(A,B;M \mid R,Z=1) \\
& = \bbI(A,B;M,R \mid Z=1) - \bbI(A,B;R \mid Z=1) \tag{by chain rule of mutual information (\cref{fact:chain-I})}\\
& = \bbI(A,B;\Gamma \mid Z=1) \\
& = \bbH(A,B \mid Z=1) - \bbH(A,B \mid \Gamma,Z=1) \\
& \ge \frac{2}{5} n\log n - \cO(n\log\log n),
\end{align*}
a contradiction.

Finally, we deal with the case of $\varepsilon > n^{-10}$.
Suppose there is a protocol with $\varepsilon \le 1/100$ and $L=o(n\log(1/\varepsilon))$.
Repeat the protocol independently for $t=\Theta(\log n / \log(1/\varepsilon))$ times in parallel.
By union bound, the majority is incorrect with probability at most \[\binom{t}{t/2} \cdot \varepsilon^{t/2} \le 2^t \cdot \varepsilon^{t/2} = (4\varepsilon)^{t/2} \le \varepsilon^{t/4} \le n^{-10}.\]
However, the total communication is only $Lt = o(n\log n)$, contradicting the case of $\varepsilon \le n^{-10}$.
The remaining case of $\varepsilon \in (1/100,1/3)$ also follows by repeating a constant number of times, concluding the proof.
\end{proof}

\subsection{Proof of \cref{lem:err-S}}
\label{sec:err-S}

This section constitutes the proof of \cref{lem:err-S}.
We take a similar approach to the proof of \cref{lem:nondet-S}.
For any $u,v \in [n]$, we extend the definition of $b_{(u,v)}$ such that it is simply the restriction of $b$ on $[n] \setminus \set{u}$ if $b(u)=v$.

The following \cref{lem:err-bp} is the analog of \cref{lem:nondet-bp,lem:nondet-W} in the randomized setting.

\begin{lemma}
\label{lem:err-bp}
For $\delta \in [0,1]$, and any product distribution $Q$ on $\cP_n \times \cP_n$ that is $\delta$-close to $1$-monochromatic, it holds that \[\Exp_{u \sim [n]} \bbH(B \mid A(u),B_{(u,A(u))}) \le \log n - \frac{\bbH(Q_\sA)}{2n} + \cO\paren*{\log\log n + \sqrt{\delta}n^3\log n},\] where $(A,B) \sim Q$, and $Q_\sA$ is the marginal of $A$.
\end{lemma}

\begin{proof}
For any fixed $u,v \in [n]$ and $b \in \cP_n$, it holds that
\begin{align}
\card{\set{b' \in \cP_n \mid b'_{(u,v)}=b_{(u,v)}}} \le n. \label{eqn:err-bp}
\end{align}
Indeed, this can be shown by straightforward case analysis similar to the proof of \cref{lem:nondet-bp}.
Also note that $\bbH(Q_\sA)/(2n) \le \cO(\log n)$.
Consequently, we can assume without loss of generality that $\sqrt{\delta} n^3 < 1$, as otherwise the $\cO(\sqrt{\delta} n^3\log n)$ term already dominates the upper bound and is trivially larger than the expectation on LHS.

Let $Q_\sB$ be the marginal of $B$, and $\cB^*$ the collection of $b \in \cB=\supp(Q_\sB)$ such that \[\Delta(Q_\sA \times \set{b}) \le \sqrt{\delta}.\]
Since $Q$ is $\delta$-close to $1$-monochromatic, by Markov's inequality, we know that \[\Pr_{b \sim Q_\sB}(b \in \cB^*) \ge 1-\sqrt{\delta}.\]
Let $I$ denote the indicator of $B \in \cB^*$, so \[\Pr(I=0) \le \sqrt{\delta} < 1.\]
As a result, we can get
\begin{align}
\bbH(B \mid A(u),B_{(u,A(u))})
& \le \bbH(B,I \mid A(u),B_{(u,A(u))}) \nonumber\\
& = \bbH(I \mid A(u),B_{(u,A(u))}) + \bbH(B \mid A(u),B_{(u,A(u))},I) \tag{by chain rule of entropy (\cref{fact:chain-H})} \nonumber\\
& \le 1 + \Pr(I=0) \cdot \bbH(B \mid A(u),B_{(u,A(u))},I=0) \nonumber\\
& \hspace{1cm} + \Pr(I=1) \cdot \bbH(B \mid A(u),B_{(u,A(u))},I=1) \nonumber\\
& \le 1 + \sqrt{\delta} \cdot \log n + \bbH(B \mid A(u),B_{(u,A(u))},I=1), \label{eqn:err-bp-I}
\end{align}
due to \cref{eqn:err-bp}.

Let $Q^*_\sB$ be the distribution $Q_\sB$ conditioned on $I=1$, so it suffices to upper bound \[\Exp_{u \sim [n]} \bbH(B^* \mid A(u),B^*_{(u,A(u))}),\] where $(A,B^*) \sim (Q_\sA,Q^*_\sB)$.
Define the neighborhood of $b$ inside $\cB^*$ to be \[\cB^*(b) = \set{b' \in \cB^* \mid \exists u,v \in [n], b'_{(u,v)}=b_{(u,v)}}.\]
Note that $\cB^*(b)$ always contains $b$ itself and that its size is upper bounded by $n \cdot n^2 = n^3$ because of \cref{eqn:err-bp}.
Also define the neighborhood of $b$ inside $\cA=\supp(Q_\sA)$ to be \[\cA^*(b) = \set{a \in \cA \mid \forall b' \in \cB^*(b), (a,b') \in \bconn_N^{-1}(1)}.\]
Let $J$ denote the indicator of $A \in \cA^*(B^*)$.
By union bound, we have that \[\Pr(J=0 \mid B^*=b) \le \sqrt{\delta} \cdot \card{\cB^*(b)} \le \sqrt{\delta} n^3 < 1.\]
Then, it must hold that
\begin{align}
\bbH(B^* \mid A(u),B^*_{(u,A(u))})
& \le \bbH(B^*,J \mid A(u),B^*_{(u,A(u))}) \nonumber\\
& = \bbH(J \mid A(u),B^*_{(u,A(u))}) + \bbH(B^* \mid A(u),B^*_{(u,A(u))},J) \tag{by chain rule of entropy (\cref{fact:chain-H})} \nonumber\\
& \le 1 + \Pr(J=0) \cdot \bbH(B^* \mid A(u),B^*_{(u,A(u))},J=0) \nonumber\\
& \hspace{1cm} + \Pr(J=1) \cdot \bbH(B^* \mid A(u),B^*_{(u,A(u))},J=1) \nonumber\\
& \le 1 + \sqrt{\delta} n^3 \cdot \log n + \Pr(J=1) \cdot \bbH(B^* \mid A(u),B^*_{(u,A(u))},J=1), \label{eqn:err-bp-J}
\end{align}
again due to \cref{eqn:err-bp}.

Let $Q^*_\sA(b)$ be the distribution $Q_\sA$ conditioned on $B^*=b$ and $J=1$, so we can get
\begin{align}
\bbH(Q_\sA)
& = \bbH(Q_\sA \mid B^*) \nonumber\\
& \le \bbH(Q_\sA,J \mid B^*) \nonumber\\
& = \bbH(J \mid B^*) + \bbH(Q_\sA \mid B^*,J) \tag{by chain rule of entropy (\cref{fact:chain-H})} \nonumber\\
& \le 1 + \Pr(J=0) \cdot \bbH(Q_\sA \mid B^*,J=0) + \Pr(J=1) \cdot \bbH(Q_\sA \mid B^*,J=1) \nonumber\\
& \le 1 + \sqrt{\delta} n^3 \cdot n\log n + \Pr(J=1) \cdot \bbH(Q^*_\sA(B^*) \mid B^*,J=1). \label{eqn:err-bp-Q}
\end{align}
By definition, $\cA^*(b) \times \cB^*(b)$ is a $1$-monochromatic rectangle, meaning that \cref{lem:nondet-bp,lem:nondet-W} actually can be applied\footnote{Technically, \cref{lem:nondet-W} is stated for $A$ uniformly distributed on its support $\cA$. Nevertheless, it actually holds for any distribution of $A$. The same proof applies verbatim.}.
Therefore, we also have
\begin{align}
& \Exp_{u \in [n]} \bbH(B^* \mid A(u),B^*_{(u,A(u))},J=1) \nonumber\\
& \hspace{1cm} \le \Exp_{u \in [n]} \Exp_{(a,b) \sim (Q_\sA \times Q^*_\sB) \mid J=1} \log \card{\set{b' \in \cB^* \mid b'_{(u,a(u))}=b_{(u,a(u))}}} \nonumber\\
& \hspace{1cm} = \Exp_{b \sim Q^*_\sB \mid J=1} \Exp_{u \in [n]} \Exp_{a \sim Q^*_\sA(b)} \log \card{\set{b' \in \cB^*(b) \mid b'_{(u,a(u))}=b_{(u,a(u))}}} \nonumber\\
& \hspace{1cm} \le \Exp_{b \sim Q^*_\sB \mid J=1} \paren*{\log n - \frac{\bbH(Q^*_\sA(b))}{2n} + \cO(\log\log n)} \tag{by \cref{lem:nondet-bp,lem:nondet-W}} \nonumber\\
& \hspace{1cm} = \log n - \frac{\bbH(Q^*_\sA(B^*) \mid B^*,J=1)}{2n} + \cO(\log\log n). \label{eqn:err-bp-H}
\end{align}
Plugging \cref{eqn:err-bp-H} into \cref{eqn:err-bp-I,eqn:err-bp-J} and averaging over $u \sim [n]$, we have that
\begin{align*}
& \Exp_{u \sim [n]} \bbH(B \mid A(u),B_{(u,A(u))}) \\
& \hspace{1cm} \le \log n - \frac{\Pr(J=1) \cdot \bbH(Q^*_\sA(B^*) \mid B^*,J=1)}{2n} + \cO\paren*{\log\log n + \sqrt{\delta} n^3\log n} \\
& \hspace{1cm} \le \log n - \frac{\bbH(Q_\sA)}{2n} + \cO\paren*{\log\log n + \sqrt{\delta} n^3\log n},
\end{align*}
by \cref{eqn:err-bp-Q}, concluding the proof.
\end{proof}

Now, we are ready to prove \cref{lem:err-S}.

\begin{proof}[Proof of \cref{lem:err-S}]
The proof is by induction on $n$.
The base case of $n=1$ holds because \[S_{1,\delta}=S'_{1,\delta}=0.\]

For $n>1$, fix $Q$ to be the maximizer in the definition of $S_{n,\delta}$.
Consider two independent random variables $(A,B) \sim Q$, so \[\bbH(A)+\bbH(B)=\bbH(A,B)=S_{n,\delta},\] as $Q$ is a product distribution.
Without loss of generality, assume \[\bbH(A) \ge \frac{S_{n,\delta}}{2}.\]
For any $u \in [n]$, $(A,B)$ can be equivalently represented as \[(A(u),A_{(u,A(u))},B_{(u,A(u))},\beta_{(u,A(u))}(B)),\] where $\beta_{(u,A(u))}(B) \in [\card{\set{b' \in \cB \mid b'_{(u,A(u))}=B_{(u,A(u))}}}]$ is a unique label among all $b' \in \cB$ such that $b'_{(u,A(u))}=B_{(u,A(u))}$.
Also define\footnote{Technically, $A_{(u,v)},B_{(u,v)}$ map from $[n] \setminus \set{u}$ to $[n] \setminus \set{v}$. Renaming both domain and codomain to $[n-1]$ makes them valid input to $\bconn_{N-2}$.} \[\delta_{(u,v)} = \Pr\paren*{\bconn_{N-2}(A_{(u,v)},B_{(u,v)})=0 \mid A(u)=v}.\]
It must hold that \[\Exp_{v \sim A(u)} \delta_{(u,v)} \le \delta,\] as taking restriction always preserves connectivity.
Also observe that $A_{(u,v)}$ and $B_{(u,v)}$ remain independent conditioned on any $v=A(u)$.
Consequently, we can get
\begin{align*}
\bbH(A,B)
& = \bbH\paren*{A(u),A_{(u,A(u))},B_{(u,A(u))},\beta_{(u,A(u))}(B)} \\
& = \bbH\paren*{A(u)} + \bbH\paren*{A_{(u,A(u))},B_{(u,A(u))} \mid A(u)}+ \bbH\paren*{\beta_{(u,A(u))}(B) \mid A,B_{(u,A(u))}} \tag{by chain rule of entropy (\cref{fact:chain-H})}\\
& \le \bbH\paren*{A(u)} + \Exp_{v \sim A(u)} \bbH\paren*{A_{(u,v)},B_{(u,v)} \mid A(u)=v} + \bbH\paren*{B \mid A(u),B_{(u,A(u))}} \\
& \le \bbH\paren*{A(u)} + \Exp_{v \sim A(u)} S_{n-1,\delta_{(u,v)}} + \bbH\paren*{B \mid A(u),B_{(u,A(u))}} \\
& \le \bbH\paren*{A(u)} + \Exp_{v \sim A(u)} S'_{n-1,\delta_{(u,v)}} + \bbH\paren*{B \mid A(u),B_{(u,A(u))}} \tag{by inductive hypothesis}\\
& \le \log n + S'_{n-1,\delta} + \bbH(B \mid A(u),B_{(u,A(u))}), \tag{by Jensen's inequality}
\end{align*}
where we use the fact that $S'_{n-1,\delta}$ is both concave\footnote{This is exactly the reason for introducing $S'_{n,\delta}$. Concavity does not necessarily hold for $S_{n,\delta}$ as defined.} and nondecreasing in $\delta$.
By \cref{lem:err-bp} and averaging over $u \sim [n]$, it implies that \[\bbH(A,B) \le \log n + S'_{n-1,\delta} + \log n - \frac{\bbH(A)}{2n} + \cO\paren*{\log\log n + \sqrt{\delta}n^3\log n}.\]
Since $\bbH(A) \ge S_{n,\delta}/2$ by assumption, rearranging the terms yields \[\paren*{1+\frac{1}{4n}} \cdot S_{n,\delta} \le S'_{n-1,\delta} + 2\log n + \cO\paren*{\log\log n + \sqrt{\delta}n^3\log n},\] concluding the proof by picking a sufficiently large constant $K$.
\end{proof}

%% file: const.tex
\section{Lower Bound for Connectivity with Constant Error}
\label{sec:const}

We prove round-dependent randomized lower bounds for connectivity with constant error in this section.
As will be seen in \cref{thm:property}, the reduction can be easily adapted to prove similar lower bounds for deciding other fundamental graph properties.

\begin{theorem}
\label{thm:const}
For even $N>1$, and $r=o(\log N)$, any $r$-round randomized protocol for $\bconn_N$ with constant error probability less than $1/2$ requires at least $\bfR^{(r)}(\bconn_N) = \Omega(N/r \cdot \log N)$ communication.
\end{theorem}

\begin{theorem}
\label{thm:property}
For $N \ge 1$, and $r=o(\log N)$, any $r$-round randomized protocol requires at least $\Omega(N/r \cdot \log N)$ communication to decide with constant error probability less than $1/2$ if an $N$-vertex graph is 
\begin{inparaenum}[(1)]
    \item bipartite;
    \item cycle-free;
    \item Eulerian;
    \item Hamiltonian;
    \item planar.
\end{inparaenum}
\end{theorem}

Another direct consequence of \cref{thm:const} is a multi-pass streaming lower bound for connectivity, as shown in \cref{cor:stream}.
Since global minimum cuts can only be harder, it also fully resolves the optimality of the two-pass algorithm in \cite{AD21}, even if any constant number of passes is allowed.

\begin{corollary}
\label{cor:stream}
For $N \ge 1$, and $p=o(\log N)$, any $p$-pass randomized streaming algorithm that decides connectivity in $N$-vertex graphs with constant error probability less than $1/2$, requires at least $\Omega(N/p^2 \cdot \log N)$ space.
\end{corollary}

\begin{proof}
By standard reductions, any $p$-pass $S$-space streaming algorithm can be simulated by a $(2p-1)$-round communication protocol with $O(pS)$ communication.
Consequently, the corollary follows from \cref{thm:const}.
\end{proof}

\begin{remark}
We remark that the $2$-round $\Omega(N\log N)$ communication lower bound for connectivity in \cite{RRV26} alone does not automatically imply optimality of \cite{AD21}, as the standard simulation of $2$-pass streaming algorithms uses $2 \cdot 2 - 1 = 3$ communication rounds.
\end{remark}

Proofs of \cref{thm:const,thm:property} utilize the following communication lower bound for $\prob$, which will be proved in \cref{sec:prob}.

\begin{restatable*}{theorem}{thmprob}
\label{thm:prob}
For $n \ge 1$, and $r=o(\log n)$, any $r$-round protocol for $\prob_{n,r}$ with constant error probability less than $1/2$ requires at least $\bfR^{(r)}(\prob_{n,r}) = \Omega(n \log n)$ communication.
\end{restatable*}

In \cref{sec:const-conn}, we present the reduction from $\prob$.
\cref{thm:property} is shown in \cref{sec:const-other}.

\subsection{Proof of \cref{thm:const}}
\label{sec:const-conn}

This section constitutes the proof of \cref{thm:const}.
Fix sufficiently large even $N$, and $r=o(\log N)$.
We first prove the lower bound for $\conn_N$.

At a high level, the input permutations $A_i,B_i$ of $\prob_{n,r}$ are encoded by perfect matchings between consecutive layers in the constructed graph instance.
The two bit vectors $X,Y$ decide the structure of interleaved special layers so that the number of cycles is determined by the output of $\prob_{n,r}$.
Similar approaches have been taken to provide reductions from other pointer-chasing-type communication problems; see for example \cite{AN21,RRV26}.

The graph instance $G=(V,E)$ is constructed as follows. 

\paragraph{Vertices.}

$V$ consists of $4\ell+1$ layers $V_0,\dots,V_{4\ell}$, each with $4n$ vertices, where $\ell=2r+1$.
The total number of vertices is $N = (4\ell+1) \cdot 4n = \Theta(rn)$.
Each vertex is indexed by $(i,j,k) \in [0,4\ell] \times [n] \times [4]$, so vertices in the same layer are partitioned into groups of size $4$.

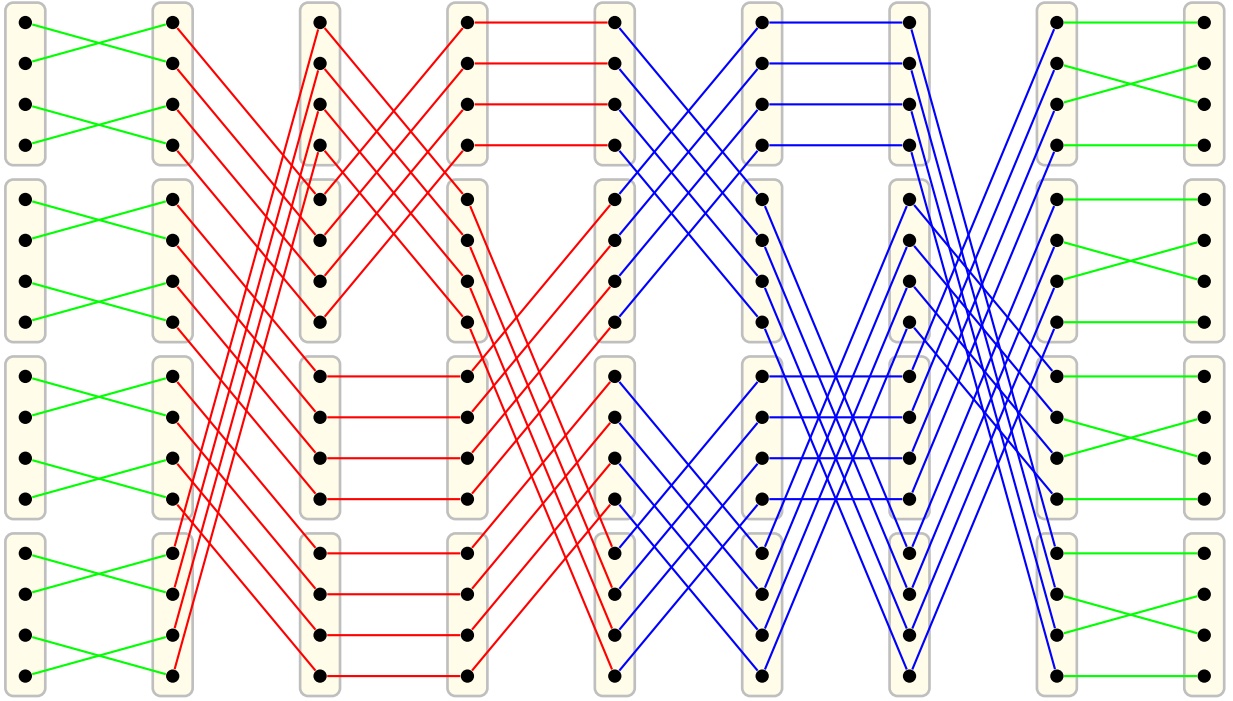
\begin{figure}[tp]
    \centering
    \input{fig/reduce}
    \caption{An illustration of the first $\ell+1$ layers of permutation and product edges in $G$, with $n=4,r=3$ and thus $\ell=7$. Each yellow group represents a group of $4$ vertices in the same layer. Red edges correspond to permutations $A_1=(1\ 2\ 3\ 4)$, $A_2=(1\ 2)$, $A_3=(2\ 4\ 3)$. Blue edges correspond to inverses of permutations $B_1=(1\ 3\ 2\ 4)$, $B_2=(2\ 4)$, $B_3=(1\ 2)(3\ 4)$ in reverse order. Green edges are product edges.}
    \label{fig:reduce}
\end{figure}

\paragraph{Edges.}

$E$ consists of different types of edges.
All edges are between two consecutive layers, or $V_{4\ell}$ and $V_0$.
The type of an edge is determined solely by the layers that it connects.
All edges connecting the same two layers are owned by the same party.

First consider \emph{permutation edges}.
For $i \in [r]$ and $j \in [n]$, add the identity matching (i.e., vertices with the same last index are connected) between each of the following groups:
\begin{enumerate}
    \item $(i,j,[4])$ and $(i+1,A_i(j),[4])$;
    \item $(r+i,j,[4])$ and $(r+i+1,B_{r-i+1}^{-1}(j),[4])$;
    \item $(\ell+i,j,[4])$ and $(\ell+i+1,B_i(j),[4])$;
    \item $(\ell+r+i,j,[4])$ and $(\ell+r+i+1,A^{-1}_{r-i+1}(j),[4])$;
    \item $(2\ell+i,j,[4])$ and $(2\ell+i+1,A_i(j),[4])$;
    \item $(2\ell+r+i,j,[4])$ and $(2\ell+r+i+1,B_{r-i+1}^{-1}(j),[4])$;
    \item $(3\ell+i,j,[4])$ and $(3\ell+i+1,B_i(j),[4])$;
    \item $(3\ell+r+i,j,[4])$ and $(3\ell+r+i+1,A^{-1}_{r-i+1}(j),[4])$.
\end{enumerate}
Essentially, the $4\ell$ layers of edges are equally partitioned into $4$ blocks.
In each block, all but the first layer of edges are permutation edges, with each layer corresponding to one of the $2r$ permutations.
The first and the third blocks encode the permutations in the order of $A_1,\dots,A_r,B_r^{-1},\dots,B_1^{-1}$, while the second and the fourth blocks encode the inverse permutations in the reverse order, namely $B_1,\dots,B_r,A_r^{-1},\dots,A_1^{-1}$.
Therefore, while preserving the last index, the net effect of permutation edges on the middle index is simply permutations $C,C^{-1},C,C^{-1}$ in the four blocks, respectively, where
\begin{align}
C = (B_r \cdots B_1)^{-1} A_r \cdots A_1. \label{eqn:reduce-perm}
\end{align}

The first layer in each block is reserved for \emph{product edges}.
Specifically, let $\sigma_\sX=(1\ 2)(3\ 4)$ and $\sigma_\sY=(2\ 3)$ be $4$-permutations.
For $j \in [n]$, add the following product edges:
\begin{enumerate}
    \item Edges corresponding to $(\sigma_\sX)^{X_j}$ between $(0,j,[4])$ and $(1,j,[4])$;
    \item Edges corresponding to $(\sigma_\sY)^{Y_j}$ between $(\ell,j,[4])$ and $(\ell+1,j,[4])$;
    \item Edges corresponding to $(\sigma_\sX)^{-X_j}=(\sigma_\sX)^{X_j}$ between $(2\ell,j,[4])$ and $(2\ell+1,j,[4])$;
    \item Edges corresponding to $(\sigma_\sY)^{-Y_j}=(\sigma_\sY)^{Y_j}$ between $(3\ell,j,[4])$ and $(3\ell+1,j,[4])$.
\end{enumerate}
Observe that product edges always preserve the middle index.
More important, for any bits $x,y \in \bits$, the net effect of product edges on the last index is exactly
\begin{align}
(\sigma_\sY)^{-y}(\sigma_\sX)^{-x}(\sigma_\sY)^y(\sigma_\sX)^x = \tau^{x \cdot y}, \label{eqn:reduce-product}
\end{align}
where $\tau=(1\ 4)(2\ 3)$.
Indeed, if either $x$ or $y$ is $0$, then it trivially evaluates to the identity permutation.
Otherwise, $\sigma_\sY \sigma_\sX \sigma_\sY \sigma_\sX = \tau$.
Note that $\tau^2$ is simply the identity permutation.

Lastly, there are also \emph{fixed edges}.
For now, the fixed edges only connect $V_{4\ell}$ and $V_0$.
We will insert a few more layers of fixed edges for $\bconn_N$.
Concretely, for $j \in [n-1]$, add the identity matching between $(4\ell,j,[4])$ and $(0,j+1,[4])$.
In the meantime, add the edges corresponding to the cyclic $4$-permutation $\pi=(1\ 2\ 3\ 4)$ between $(4\ell,n,[4])$ and $(0,1,[4])$.

This completes the construction of $G$.
See \cref{fig:reduce} for an illustration of (part of) permutation and product edges.
Note that it only depicts the first block of permutation edges and two adjacent layers of product edges.
It does not contain the fixed edges.

\cref{lem:reduce} provides an exact characterization of the structure of $G$.

\begin{lemma}
\label{lem:reduce}
$G$ consists of a single $N$-cycle if $\ang{A_r \cdots A_1 X, B_r \cdots B_1 Y} = 0$, and of three disjoint cycles of length $N/2,N/4,N/4$ if $\ang{A_r \cdots A_1 X, B_r \cdots B_1 Y} = 1$.
\end{lemma}

\begin{proof}
The observations above, especially \cref{eqn:reduce-perm,eqn:reduce-product}, establish that every vertex $(0,j,k) \in V_0$ is connected via a unique and disjoint path to
\begin{enumerate}
    \item $(\ell,C(j),(\sigma_\sX)^{X_j}(k)) \in V_\ell$ by the first block;
    \item $(2\ell,j,(\sigma_\sY)^{Y_{C(j)}}(\sigma_\sX)^{X_j}(k)) \in V_{2\ell}$ by the second block;
    \item $(3\ell,C(j),(\sigma_\sX)^{-X_j}(\sigma_\sY)^{Y_{C(j)}}(\sigma_\sX)^{X_j}(k)) \in V_{3\ell}$ by the third block;
    \item $(4\ell,j,(\sigma_\sY)^{-Y_{C(j)}}(\sigma_\sX)^{-X_j}(\sigma_\sY)^{Y_{C(j)}}(\sigma_\sX)^{X_j}(k)) \in V_{4\ell}$ by the fourth block.
\end{enumerate}
By \cref{eqn:reduce-product}, the net effect is exactly $\tau^{X_j \cdot Y_{C(j)}}$ on the last index.

Furthermore, observe that the identity matchings added as fixed edges essentially compose the net effect for every $j \in [n]$ in sequence.
In more detail, excluding the fixed edges corresponding to $\pi$, $G$ is the union of four disjoint paths, each with $N/4$ vertices, from $(0,1,k)$ to $(4\ell,n,\tau^O(k))$ for $k \in [4]$, where \[O=\sum_{j \in [n]} X_j \cdot Y_{C(j)} = \ang{A_r \cdots A_1 X, B_r \cdots B_1 Y} \pmod{2}.\]
To see the last step, note that for $j' \in [n]$, the $j'$-th bit of $A_r \cdots A_1 X$ is the $(A_r \cdots A_1)^{-1}(j')$-th bit of $X$, while the $j'$-th bit of $B_r \cdots B_1 Y$ is the $(B_r \cdots B_1)^{-1}(j')$-th bit of $Y$.
Plugging in $j=(A_r \cdots A_1)^{-1}(j')$ yields the desired equality by \cref{eqn:reduce-perm}.

Finally, the proof is completed by observing $\pi\tau=(2\ 4)$.
\end{proof}

Now, we are ready to prove \cref{thm:const}.

\begin{proof}[Proof of \cref{thm:const}]
Given any $\prob_{n,r}$ input $(X,Y,A_{\le r},B_{\le r})$, all edges of the graph instance $G$ can be constructed by either Alice or Bob.
Indeed, permutation edges corresponding to $A_i,B_i$ (or $A_i^{-1},B_i^{-1}$) are owned by Alice for odd i and are owned by Bob for even $i$.
Bob also owns both product edges and fixed edges.
Therefore, \cref{lem:reduce} implies that any protocol for $\conn_N$ also solves $\prob_{n,r}$ with the same number of rounds, error probability, and communication.
The lower bound for $\conn_N$ then follows from the lower bound for $\prob_{n,r}$ (\cref{thm:prob}) as $n=\Theta(N/r)$ and thus $r=o(\log n)$.

For $\bconn_N$, we additionally insert a constant number of layers of fixed edges so that the graph instance $G$ is bipartite and that Alice and Bob own alternating layers of edges.
Concretely, the current construction already has alternating ownership except that:
\begin{enumerate}
    \item Layers of edges corresponding to $A_r,B^{-1}_r$ (or $B_r,A_r^{-1}$) are owned by the same party.
        To fix this, simply insert an additional layer of fixed edges in between, corresponding to the identity permutation, that is owned by the opposite party.
    \item Bob owns both edges between $V_0,V_1$ and the ones between $V_{4\ell},V_0$.
        This can also be fixed by inserting an additional layer of fixed edges after $V_0$, corresponding to the identity permutation, that is owned by Alice.
\end{enumerate}
Then, the modified graph has a balanced bipartition such that both Alice and Bob own a perfect matching across the bipartition.
This concludes the proof\footnote{Technically, we only construct $G$ for $N$ of the form $(4\ell+1+5) \cdot 4n$. For any other even $N$, simply let $n = \floor{N/(16\ell+24)}$. Splitting the fixed edge $((4\ell,n,1),(0,1,2))$ into a path of length $N - (16\ell+24) \cdot n + 1$ ensures exactly $N$ vertices. Assign the subdivided edges to the two parties alternately. No analysis is affected.} for $\bconn_N$ since the number of layers increases only by a constant (exactly $4+1=5$, to be precise).
\end{proof}

\subsection{Lower Bounds for Other Graph Properties}
\label{sec:const-other}

\cref{thm:property} can be shown via standard modifications to the reduction in \cref{sec:const-conn}.

\begin{proof}[Proof of \cref{thm:property}]
The proof of \cref{thm:const} actually constructs a bipartite graph $G$ consisting of either a single even cycle or three disjoint even cycles such that the two fixed edges $((4\ell,n,4),(0,1,1))$, $((4\ell,n,2),(0,1,3))$ are contained in different cycles.

The lower bound for bipartiteness follows by subdividing each edge above into a path of length $2$.
Indeed, if there is a single even cycle, $G$ remains bipartite, as the length of the cycle increases by $2$.
Otherwise, $G$ contains two odd cycles and is not bipartite anymore.

For cycle-freeness, remove the above edges.
In the case of a single cycle, removing the above edges breaks the cycle.
Otherwise, $G$ still contains one cycle.

A graph is Eulerian if and only if it is connected and every vertex has an even degree.
Since every vertex has degree exactly $2$ in $G$, it is Eulerian if and only if it is connected.

Similarly, $G$ is Hamiltonian if and only if it is connected.

As to planarity, add three new vertices $u_1,u_2,u_3$, and construct a $5$-clique among $u_1,u_2,u_3$ and $(0,1,1),(0,1,3)$, except for the edge between $(0,1,1)$ and $(0,1,3)$.
If $G$ is a single cycle, then the cycle completes a $K_5$ minor and thus the resulting graph is not planar.
Otherwise, since $(0,1,1)$ and $(0,1,3)$ are in two disconnected cycles, the resulting graph remains planar.

This concludes the proof as the lower bound for each of the above graph properties follows from a similar reduction from $\prob_{n,r}$.
\end{proof}

%% file: fig/reduce.tex
\begin{tikzpicture}

\tikzset{
    layer/.style={rectangle,rounded corners=3pt,draw,gray!50,line width=1pt,fill=yellow!10,inner sep=4pt},
    vertex/.style={circle,fill,minimum size=5pt,inner sep=0pt}
}

\tikzmath{
    \n=4;\r=3;\l=2*\r+1;\lp=\l+1;
    \a{1,1}=2;\a{1,2}=3;\a{1,3}=4;\a{1,4}=1;
    \a{2,1}=2;\a{2,2}=1;\a{2,3}=3;\a{2,4}=4;
    \a{3,1}=1;\a{3,2}=4;\a{3,3}=2;\a{3,4}=3;
    \b{1,1}=3;\b{1,2}=4;\b{1,3}=2;\b{1,4}=1;
    \b{2,1}=1;\b{2,2}=4;\b{2,3}=3;\b{2,4}=2;
    \b{3,1}=2;\b{3,2}=1;\b{3,3}=4;\b{3,4}=3;
    \x{1}=2;\x{2}=1;\x{3}=4;\x{4}=3;
    \y{1}=1;\y{2}=3;\y{3}=2;\y{4}=4;
}

\begin{scope}[local bounding box=b01]
    \node[vertex] (v011) {};
    \foreach \k in {2,...,4} {
        \pgfmathtruncatemacro{\kp}{\k-1}
        \node[vertex] (v01\k) [below=10pt of v01\kp] {};
    }
\end{scope}

\foreach \j in {2,...,\n} {
    \pgfmathtruncatemacro{\jp}{\j-1}
    \begin{scope}[local bounding box=b0\j]
        \node[vertex] (v0\j1) [below=15pt of v0\jp4] {};
        \foreach \k in {2,...,4} {
            \pgfmathtruncatemacro{\kp}{\k-1}
            \node[vertex] (v0\j\k) [below=10pt of v0\j\kp] {};
        }
    \end{scope}
}

\foreach \j in {1,...,\n} {
    \begin{scope}[on background layer]
        \draw[layer] ($(b0\j.south west) - (5pt,5pt)$) rectangle ($(b0\j.north east) + (5pt,5pt)$);
    \end{scope}
}

\foreach \i in {1,...,\lp} {
    \pgfmathtruncatemacro{\ip}{\i-1}
    \begin{scope}[local bounding box=b\i1]
        \node[vertex] (v\i11) [right=50pt of v\ip11] {};
        \foreach \k in {2,...,4} {
            \pgfmathtruncatemacro{\kp}{\k-1}
            \node[vertex] (v\i1\k) [below=10pt of v\i1\kp] {};
        }
    \end{scope}

    \foreach \j in {2,...,\n} {
        \pgfmathtruncatemacro{\jp}{\j-1}
        \begin{scope}[local bounding box=b\i\j]
            \node[vertex] (v\i\j1) [below=15pt of v\i\jp4] {};
            \foreach \k in {2,...,4} {
                \pgfmathtruncatemacro{\kp}{\k-1}
                \node[vertex] (v\i\j\k) [below=10pt of v\i\j\kp] {};
            }
        \end{scope}
    }

    \foreach \j in {1,...,\n} {
        \begin{scope}[on background layer]
            \draw[layer] ($(b\i\j.south west) - (5pt,5pt)$) rectangle ($(b\i\j.north east) + (5pt,5pt)$);
        \end{scope}
    }
}

\foreach \i in {1,...,\r} {
    \pgfmathtruncatemacro{\in}{\i+1}
    \pgfmathtruncatemacro{\ii}{\l-\i+1}
    \pgfmathtruncatemacro{\iip}{\ii-1}
    \foreach \j in {1,...,\n} {
        \foreach \k in {1,...,4} {
            \draw[thick,red] (v\i\j\k) -- (v\in\a{\i,\j}\k);
            \draw[thick,blue] (v\ii\j\k) -- (v\iip\b{\i,\j}\k);
        }
    }
}

\pgfmathtruncatemacro{\i}{\l}
\pgfmathtruncatemacro{\in}{\lp}
\foreach \j in {1,...,\n} {
    \foreach \k in {1,...,4} {
        \draw[thick,green] (v0\j\k) -- (v1\j\x{\k});
        \draw[thick,green] (v\i\j\k) -- (v\in\j\y{\k});
    }
}

\end{tikzpicture}

%% file: prob.tex
\section{Lower Bound for $\prob$}
\label{sec:prob}

In this section, we prove the communication lower bound for $\prob$, as promised in \cref{sec:const}.

\thmprob

\begin{remark}
The requirement of $r=o(\log n)$ can actually be relaxed to many more rounds.
We remark that \cref{thm:prob} already suffices to prove the lower bound for connectivity in \cref{sec:const}, and larger $r$ does not help obtain a stronger lower bound.
Nevertheless, we extend the lower bound in \cref{sec:prob-sublinear}, as it may be of independent interest.
\end{remark}

The rest of this section constitutes the proof of \cref{thm:prob}.
We first prove a weaker version of \cref{thm:prob} in \cref{sec:prob-weak}, showing that $\bfR^{(r)}(\prob_{n,r}) = \Omega(n/r \cdot \log n)$.
It already demonstrates many of the crucial ideas for the full proof.
As to graph connectivity, by the same reduction as in \cref{sec:const}, the weaker version is sufficient for proving the optimal lower bound of $\bfR^{(r)}(\conn_N) = \Omega(N \log N)$ for constant rounds.
However, it yields a worse $1/r^2$ dependency on $r$ for superconstant rounds, as compared to \cref{thm:const}.

In \cref{sec:prob-full}, we give a tighter analysis of how much randomness in the input permutations can actually be lost during the protocol.
This completes the full proof of \cref{thm:prob}.

\subsection{A Weaker Lower Bound}
\label{sec:prob-weak}

Fix sufficiently large $n$, and $r=o(\log n)$.
Suppose that there exists an $r$-round protocol with constant error probability $\varepsilon < 1/2$ and communication $L=o(n/r \cdot \log n)$.
We can actually assume $\varepsilon \le 1/3$ by repetition.
We further assume that the protocol is deterministic and has error probability $\varepsilon$ on uniformly random input.

The high level plan is to show that the pair of initially independent and uniformly random bit vectors $(X,Y)$ retains almost full entropy through the protocol, and eventually, after taking inner product, leads to one random output bit with hardly any bias.
More specifically, let $(W_0,Z_0)=(X,Y)$ and recursively define $(W_i,Z_i)=(A_i W_{i-1}, B_i Z_{i-1})$ for $i \in [r]$.
In round $i$, the $i$-th pair of permutations $(A_i,B_i)$ should still remain close to random because the total communication is limited.
Then, intuitively, $(W_i,Z_i)$ should also have enough randomness, provided that $(W_{i-1},Z_{i-1})$ is sufficiently random to start with.

\cref{lem:prob-D} formalizes the above intuition in the ideal case where $(W_{i-1},Z_{i-1})$ were uniformly random.
With $Q$ being the current distribution of $(A_i,B_i)$, it essentially upper bounds the ``distance'' of $(W_i,Z_i)$ from being uniformly random, in terms of the ``distance'' of $(A_i,B_i)$ from being uniformly random (although the ``distances'' are measured by \renyi{} divergence of different orders).

\begin{lemma}
\label{lem:prob-D}
For any distribution $Q$ on $\cP_n \times \cP_n$, it holds that \[\Exp_{w,z \sim \bits^n} 2^{\bbD_2(Q^{(w,z)})} \le 2^{\frac{\bbD_\infty(Q)+\cO(n)}{\log n}},\]
where $Q^{(w,z)}$ is the distribution of $(aw,bz) \in \bits^{2n}$ when $(a,b) \sim Q$.
\end{lemma}

The proof of \cref{lem:prob-D} uses the following technical lemma.

\begin{claim}
\label{clm:cyc}
For $m \ge 1$, permutations $\sigma_1,\dots,\sigma_m \in \cP_n$, and any distribution $Q$ on $(\cP_n)^m$, it holds that \[\Exp_{(\sigma'_1,\dots,\sigma'_m) \sim Q} 2^{\sum_{i \in [m]} \cyc(\sigma_i^{-1}\sigma'_i)} \le 2^{\frac{\bbD_\infty(Q)+\cO(mn)}{\log n}}.\]
\end{claim}

\begin{proof}
Observe that $\sigma'_i$ can always be decomposed as $\sigma'_i=\sigma_i\sigma''_i$ for a unique $\sigma''_i \in \cP_n$.
It then holds that $\sigma_i^{-1}\sigma'_i=\sigma''_i$.
Moreover, $\sigma''_i$ jointly follow the same distribution $Q$, up to renaming induced by $\sigma_i$.
Thus, we can assume without loss of generality that all $\sigma_i$ are simply the identity permutation.

Let $t=(\bbD_\infty(Q)+Kmn)/\log n$, for some sufficiently large constant $K$.
The contribution to the expectation is at most $2^t$, for any $\sigma'_1,\dots,\sigma'_m \in \cP_n$ such that $\sum_{i \in [m]} \cyc(\sigma'_i) \le t$.
In the meantime, we claim that the total number of possible combinations of $\sigma'_1,\dots,\sigma'_m \in \cP_n$ such that $\sum_{i \in [m]} \cyc(\sigma'_i) > t$ is at most $2^{mn}n^{mn-t}$.
Assume the claim for now.
Note that $\sum_{i \in [m]} \cyc(\sigma'_i) \le mn$ always holds by definition.
As a result, by definition of min-entropy, the total contribution in this case is at most
\[2^{mn} \cdot 2^{mn}n^{mn-t} \cdot 2^{-\bbH_\infty(Q)} = 2^{2mn + mn\log n - t\log n - m\log(n!) + \bbD_\infty(Q)} = o(1),\]
where we use Stirling's approximation.
The proof concludes once we show the claim.

To see the claim, consider constructing the unique cycle decomposition representation of $\sigma'_i$ by writing its elements in increasing order.
If the current element $j \in [n]$ will be the smallest one in the cycle containing it, then it can only be appended to the end of the current partial representation, as a $1$-cycle for now.
Otherwise, it can be inserted into existing cycles by being appended to one of the $j-1$ previous elements.
Overall, there are at most $mn-\sum_{i \in [m]} \cyc(\sigma'_i) < mn-t$ such elements in total across all $\sigma'_i$, each with at most $n$ options, whereas all other elements have only a single option.
This concludes the proof as there are no more than $2^{mn}$ ways to pick such elements.
\end{proof}

Now, we are ready to prove \cref{lem:prob-D}.

\begin{proof}[Proof of \cref{lem:prob-D}]
By definition of collision entropy,
\begin{align*}
\Exp_{w,z \sim \bits^n} 2^{\bbD_2(Q^{(w,z)})}
& = 2^{2n} \cdot \Exp_{w,z \sim \bits^n} 2^{-\bbH_2(Q^{(w,z)})} \\
& = 2^{2n} \cdot \Exp_{w,z \sim \bits^n} \Exp_{(u,v),(u',v') \sim Q^{(w,z)}} \1\paren*{(u,v)=(u',v')} \\
& = 2^{2n} \cdot \Exp_{(a,b),(a',b') \sim Q} \Exp_{w,z \sim \bits^n} \1\paren*{(aw,bz)=(a'w,b'z)}.
\end{align*}
Observe that the condition $aw=a'w$ is equivalent to $a^{-1}a'w=w$.
Since $a^{-1}a'$ is also a permutation, this happens if and only if $w$ is constant when restricted to every cycle in the decomposition of $a^{-1}a'$.
For uniformly random $w$, such an event occurs with probability exactly $2^{\cyc(a^{-1}a')-n}$.
A similar argument applies to $bz=b'z$.
As the two parts are independent for fixed $(a,b),(a',b')$,  we can get \[\Exp_{w,z \sim \bits^n} \1\paren*{(aw,bz)=(a'w,b'z)} = 2^{\cyc(a^{-1}a')+\cyc(b^{-1}b')-2n}.\]
Continuing, we have that
\begin{align*}
\Exp_{w,z \sim \bits^n} 2^{\bbD_2(Q^{(w,z)})}
& = \Exp_{(a,b),(a',b') \sim Q} 2^{\cyc(a^{-1}a')+\cyc(b^{-1}b')} \\
& \le \max_{a,b \in \cP_n} \Exp_{(a',b') \sim Q} 2^{\cyc(a^{-1}a')+\cyc(b^{-1}b')}.
\end{align*}
The lemma then follows from \cref{clm:cyc} for $m=2$.
\end{proof}

In reality, we can only expect $(W_{i-1},Z_{i-1})$ to be close to random.
To this end, for $i \in [r]$, let $M_i$ denote the $i$-th message.
Define\footnote{Assume each message is of fixed length. The full proof in \cref{sec:prob-full} does not assume this. Specifically, any prefix-free message of varying length is allowed so long as the total message length is bounded.} $L_i = \sum_{j \in [i]} \card{M_j}$ and $\Phi_i = \sum_{j \in [i]} (L_j+Kn)/\log n$, for some sufficiently large constant $K$.
Also let $\Gamma_i=(X,Y,A_{<i},B_{<i},M_{\le i})$ denote the partial transcript up to round $i$, together with the input variables that we are going to additionally reveal.
We emphasize that $\Gamma_i$ is only for analysis purposes and is not revealed to the parties during execution of the protocol.
Also note that it includes the $i$-th message, but not the $i$-th pair of permutations.
In particular, $(w_{i-1},z_{i-1})$ is fixed once $\Gamma_i=\gamma_i$ is revealed, whereas $(W_i,Z_i)$ is not.
For convenience, define\footnote{$\Gamma_0$ excludes $(X,Y)$ due to the definition of $\cE_1,\cE'_1$ below. This will be useful in later proofs.} $L_0=\Phi_0=0$ and $\Gamma_0=\varnothing$.

Given\footnote{Note that $\Gamma_{i-1}$ excludes $A_{i-1},B_{i-1}$ (or $X,Y$ for $i=1$), so it does not fix $(W_{i-1},Z_{i-1})$.} $(\Gamma_{i-1},W_{i-1},Z_{i-1},M_i)=(\gamma_{i-1},w_{i-1},z_{i-1},m_i)$, first define $\cE_i$ to be the event that $\bbD_2(W_i,Z_i \mid \gamma_{i-1},w_{i-1},z_{i-1},m_i) > \Phi_i$.
For convenience, define $\cE_0=\varnothing$ to be a never happening event.
We want to show that $\cE_i$ rarely happens.
That is, $(W_i,Z_i)=(A_i w_{i-1}, B_i z_{i-1})$ should still be sufficiently random.

In order to prove this, also define $\cE'_i$ to be the event that $\bbD_\infty(A_i,B_i \mid \gamma_{i-1},m_i) > L_i+n$.
We first show in \cref{lem:prob-Ep} that with extremely high probability, a typical partial transcript of total length $L_i$ should not cause $\cE'_i$ to happen.

\begin{lemma}
\label{lem:prob-Ep}
For $i \in [r]$, it holds that \[\Pr(\cE'_i) \le 2^{-n}.\]
\end{lemma}

\begin{proof}
We prove the lemma in the case of odd $i$.
The case of even $i$ is similar.

At the beginning of the protocol, all permutations are independent and uniformly random.
In particular, revealing $A_{<(i-1)},B_{<(i-1)}$ (and $X,Y$ for $i>1$) does not change the fact that Alice's remaining input permutations, namely $A_{i'},B_{i'}$ for odd $i' \ge i$, are still independent and uniformly random.
Let $t=\ceil{r/2}-\ceil{(i-1)/2}$ be the number of remaining pairs of input permutations for Alice.
Because of the rectangle property of deterministic communication protocols, each message $M_j$ for $j \in [i]$ either belongs to Alice, subdividing Alice's possible input into no more than $2^{\card{M_j}}$ parts, or belongs to Bob, keeping Alice's possible input the same.
Overall, the average support size of $\set{(A_{i'},B_{i'})}_{\text{odd } i' \ge i}$ after conditioning is at least $S=(n!)^{2t}2^{-L_i}$.
Moreover, the support size is smaller than $S \cdot 2^{-n}$, with probability at most $2^{-n}$ over $(\gamma_{i-1},m_i) \sim (\Gamma_{i-1},M_i)$.
To see this, the total support size of $\set{(A_{i'},B_{i'})}_{\text{odd } i' \ge i}$ for all such $(\gamma_{i-1},m_i)$ is at most \[S \cdot 2^{-n} \cdot 2^{L_i} = (n!)^{2t} \cdot 2^{-n}.\]
By the rectangle property, this leads to the $2^{-n}$ probability bound because Alice always holds a uniform distribution on the remaining support of $\set{(A_{i'},B_{i'})}_{\text{odd } i' \ge i}$.

Fix any $(\Gamma_{i-1},M_i)=(\gamma_{i-1},m_i)$ such that \[\card{\supp\paren*{\set{(A_{i'},B_{i'})}_{\text{odd } i' \ge i} \mid \gamma_{i-1},m_i}} \ge (n!)^{2t}2^{-L_i-n}.\]
For any fixed $(A_i,B_i)=(a_i,b_i)$, the total number of compatible $\set{(A_{i'},B_{i'})}_{\text{odd } i' \ge i}$ is at most $(n!)^{2(t-1)}$.
Due to uniformity, we can upper bound its marginal probability as follows: \[\Pr\paren*{(A_i,B_i)=(a_i,b_i) \mid \gamma_{i-1},m_i} \le \frac{(n!)^{2(t-1)}}{(n!)^{2t}2^{-L_i-n}} = \frac{2^{L_i+n}}{(n!)^2}.\]
Consequently, by definition of min-entropy, \[\bbD_\infty(A_i,B_i \mid \gamma_{i-1},m_i) = 2\log(n!)-\bbH_\infty(A_i,B_i \mid \gamma_{i-1},m_i) \le 2\log(n!) + \log\frac{2^{L_i+n}}{(n!)^2} = L_i+n,\]
or equivalently, $\cE'_i$ does not happen.
This concludes the proof as such $(\gamma_{i-1},m_i)$ occurs with probability at least $1-2^{-n}$, as shown above.
\end{proof}

Next, we upper bound the probability of event $\cE_i$ in \cref{lem:prob-E}.

\begin{lemma}
\label{lem:prob-E}
For $i \in [r]$, it holds that \[\Pr(\cE_i \mid \oE_{i-1},\oE'_i) \le 2^{-\frac{n}{\log n}},\] where $\oE_{i-1},\oE'_i$ are the complement events of $\cE_{i-1,}\cE'_i$, respectively.
\end{lemma}

\begin{proof}
First consider $i=1$.
$\oE_0$ trivially happens by definition.
The rectangle property of deterministic communication protocols implies that $(W_0,Z_0)=(X,Y)$ is uniformly random even conditioned on any $M_1=m_1$.
This is because $M_1$ depends only on Alice's input while $X,Y$ are owned by Bob.
Therefore, for any $m_1$ satisfying $\oE'_1$ (and $\gamma_0=\varnothing$), \cref{lem:prob-D} implies that \[\Exp_{(w_0,z_0) \mid m_1} 2^{\bbD_2(W_1,Z_1 \mid w_0,z_0,m_1)} = \Exp_{w,z \sim \bits^n} 2^{\bbD_2((A_1,B_1)^{(w,z)} \mid m_1)} \le 2^{\frac{L_1+Kn/2}{\log n}},\] for sufficiently large constant $K$.
By Markov's inequality, we can get \[\Pr_{(w_0,z_0) \mid m_1}\paren*{\cE_1} \le 2^{\frac{L_1+Kn/2}{\log n}} \cdot 2^{-\Phi_1} = 2^{-\frac{Kn/2}{\log n}} \le 2^{-\frac{n}{\log n}}.\]
Averaging over $m_1$ satisfying $\oE'_1$ yields the desired lemma statement.

For $i>1$, $\oE'_i$ implies that \[\Exp_{w,z \sim \bits^n} 2^{\bbD_2((A_i,B_i)^{(w,z)} \mid \gamma_{i-1},m_i)} \le 2^{\frac{L_i+Kn/2}{\log n}},\] for sufficiently large constant $K$.
By Markov's inequality, the total number of possible combinations of $w,z \in \bits^n$ such that $\bbD_2((A_i,B_i)^{(w,z)} \mid \gamma_{i-1},m_i) > \Phi_i$ is at most \[2^{2n} \cdot 2^{\frac{L_i+Kn/2}{\log n}} \cdot 2^{-\Phi_i} = 2^{2n-\Phi_{i-1}-\frac{Kn/2}{\log n}}.\]
At the same time, $\oE_{i-1}$ means that \[\bbD_2(W_{i-1},Z_{i-1} \mid \gamma_{i-2},w_{i-2},z_{i-2},m_{i-1}) \le \Phi_{i-1},\] or equivalently, \[\bbH_2(W_{i-1},Z_{i-1} \mid \gamma_{i-2},w_{i-2},z_{i-2},m_{i-1}) \ge 2n-\Phi_{i-1}.\]
In fact, it further implies that \[\bbH_2(W_{i-1},Z_{i-1} \mid \gamma_{i-1},m_i) \ge 2n-\Phi_{i-1}.\]
To see this, for even $i$, after revealing $(\Gamma_{i-2},W_{i-2},Z_{i-2},M_{i-1})$, both $A_{i-2},B_{i-2}$ (or $X,Y$ for $i=2$) and $M_i$ are functions of Bob's input only, while $(W_{i-1},Z_{i-1})$ depends only on Alice's input.
As a result, additional conditioning does not change its distribution, due to the rectangle property of deterministic communication protocols.
A similar argument applies to odd $i$.
Furthermore, the distribution of $(A_i,B_i)$ is conditionally independent of $(W_{i-1},Z_{i-1})$, again because of the rectangle property.
Altogether, for any $(\gamma_{i-1},m_i)$ satisfying  both $\oE_{i-1}$ and $\oE'_i$, we have that
\begin{align*}
& \Pr_{(w_{i-1},z_{i-1}) \mid (\gamma_{i-1},m_i)}(\cE_i) \\
& \hspace{1cm} = \Exp_{(w_{i-1},z_{i-1}) \mid (\gamma_{i-1},m_i)} \1\paren*{\bbD_2\paren*{(A_i,B_i)^{(w_{i-1},z_{i-1})} \mid \gamma_{i-1},m_i}>\Phi_i} \\
& \hspace{1cm} = \sum_{w,z \in \bits^n} \Pr\paren{(W_{i-1},Z_{i-1})=(w,z) \mid \gamma_{i-1},m_i} \cdot \1\paren*{\bbD_2\paren*{(A_i,B_i)^{(w,z)} \mid \gamma_{i-1},m_i}>\Phi_i} \\
& \hspace{1cm} \le \paren*{2^{-\bbH_2(W_{i-1},Z_{i-1} \mid \gamma_{i-1},m_i)} \cdot 2^{2n-\Phi_{i-1}-\frac{Kn/2}{\log n}}}^{\frac{1}{2}} \\
& \hspace{1cm} \le 2^{-\frac{Kn/4}{\log n}} \\
& \hspace{1cm} \le 2^{-\frac{n}{\log n}},
\end{align*}
by Cauchy-Schwarz inequality and definition of collision entropy.
The lemma then follows by averaging over $(\gamma_{i-1},m_i)$ satisfying both $\oE_{i-1}$ and $\oE'_i$.
\end{proof}

In the last round of the protocol, after revealing $(\Gamma_{r-1},M_r)$, which contains all messages and input permutations except $(A_{r-1},B_{r-1})$ and $(A_r,B_r)$, we can expect $(W_r,Z_r)$ to be still close to random.
In fact, the sender of $M_r$ owns $(A_r,B_r)$, uniformly distributed on its remaining support.
Meanwhile, the receiver of $M_r$ owns $(A_{r-1},B_{r-1})$, also uniformly distributed on its remaining support.
More importantly, $(A_{r-1},B_{r-1})$ and $(A_r,B_r)$ are conditionally independent, due to the rectangle property of deterministic communication protocols.
Since the receiver has to output based solely on $(A_{r-1},B_{r-1})$, it cannot predict $\ang{W_r,Z_r}$ well.

This is formalized in \cref{lem:prob-O}, where $O$ denotes the output of the protocol.

\begin{lemma}
\label{lem:prob-O}
It holds that \[\Pr(O=\ang{W_r,Z_r} \mid \oE_{r-1},\oE'_r) \le \frac{1}{2} + 2^{-\Omega(n)}.\]
\end{lemma}

\begin{proof}
First note that it suffices to upper bound the advantage \[\Pr(O=\ang{W_r,Z_r} \mid \oE_{r-1},\oE'_r)-\Pr(O \ne \ang{W_r,Z_r} \mid \oE_{r-1},\oE'_r)\] by $2^{-\Omega(n)}$.
Fix any $(\Gamma_{r-1},M_r)=(\gamma_{r-1},m_r)$.
We know that\footnote{More precisely, $(X,Y)$ should be used in place of $(A_0,B_0)$ throughout the proof if $r=1$.} $(A_{r-1},B_{r-1})$ and $(A_r,B_r)$ are independent due to the rectangle property, as discussed above.
Since the receiver of $M_r=m_r$ has to output $O$ based only on its own input and the full transcript, we also know that $O=O(A_{r-1},B_{r-1})$ is now a function of $(A_{r-1},B_{r-1})$ only and is independent of $(A_r,B_r)$, while $(W_r,Z_r)$ is fully determined by $(A_r,B_r)$ and is independent of $(A_{r-1},B_{r-1})$ conditioned on $(W_{r-1},Z_{r-1})$.
Thus, we can get
\begin{align}
& \Pr(O=\ang{W_r,Z_r} \mid \gamma_{r-1},m_r)-\Pr(O \ne \ang{W_r,Z_r} \mid \gamma_{r-1},m_r) \nonumber\\
& \hspace{1cm} = \Exp_{(w_r,z_r,o) \mid (\gamma_{r-1},m_r)} (-1)^{o+\ang{w_r,z_r}} \nonumber\\
& \hspace{1cm} = \sum_{w,z \in \bits^n} f(w,z) g(w,z) \nonumber\\
& \hspace{1cm} \le \paren*{\sum_{w,z \in \bits^n} f(w,z)^2}^{\frac{1}{2}} \paren*{\sum_{w,z \in \bits^n} g(w,z)^2}^{\frac{1}{2}}, \label{eqn:prob-O}
\end{align}
by Cauchy-Schwarz inequality, where \[f(w,z) = \Exp_{(a_{r-1},b_{r-1}) \mid \gamma_{r-1},m_r} \1\paren*{(W_{r-1},Z_{r-1})=(w,z)} \cdot (-1)^{O(a_{r-1},b_{r-1})},\] and \[g(w,z) = \Exp_{(w_r,z_r) \sim (A_r,B_r)^{(w,z)} \mid (\gamma_{r-1},m_r)} (-1)^{\ang{w_r,z_r}}.\]

Observe that\footnote{For convenience, we use the convention $m_0=\varnothing$. Similarly for any variable with negative subscript.}
\begin{align*}
\abs{f(w,z)}
& \le \Pr\paren*{(W_{r-1},Z_{r-1})=(w,z) \mid \gamma_{r-1},m_r} \\
& = \Pr\paren*{(W_{r-1},Z_{r-1})=(w,z) \mid \gamma_{r-2},w_{r-2},z_{r-2},m_{r-1}},
\end{align*}
because $(W_{r-1},Z_{r-1})$ is conditionally independent of $(A_{r-2},B_{r-2})$ and $M_r$, again by the rectangle property.
Consequently, by definition of collision entropy, we have that
\begin{align}
\sum_{w,z \in \bits^n} f(w,z)^2
& \le 2^{-\bbH_2(W_{r-1},Z_{r-1} \mid \gamma_{r-2},w_{r-2},z_{r-2},m_{r-1})} \nonumber\\
& = 2^{-2n+\bbD_2(W_{r-1},Z_{r-1} \mid \gamma_{r-2},w_{r-2},z_{r-2},m_{r-1})} \nonumber\\
& \le 2^{-2n+\Phi_{r-1}}, \label{eqn:prob-O-f}
\end{align}
for any $(\gamma_{r-2},w_{r-2},z_{r-2},m_{r-1})$ satisfying $\oE_{r-1}$.

Furthermore, let $Q$ denote the distribution of $(A_r,B_r)$ conditioned on $(\Gamma_{r-1},M_r)=(\gamma_{r-1},m_r)$, for notational convenience.
Define $Q'$ to be the distribution of $a^{-1}b$ for $(a,b) \sim Q$.
By definition of min-entropy, we have that \[2^{\bbD_\infty(Q')} = n! \cdot 2^{-\bbH_\infty(Q')} \le n! \cdot n! \cdot 2^{-\bbH_\infty(Q)} = 2^{\bbD_\infty(Q)},\] because at most $n!$ pairs of $(a,b)$ can map to the same $a^{-1}b$.
Equivalently, it must be that $\bbD_\infty(Q') \le \bbD_\infty(Q)$.
Then, we can get
\begin{align}
\sum_{w,z \in \bits^n} g(w,z)^2
& = \sum_{w,z \in \bits^n} \paren*{\Exp_{(a,b) \sim Q} (-1)^{\ang{aw,bz}}} \paren*{\Exp_{(a',b') \sim Q} (-1)^{\ang{a'w,b'z}}} \nonumber\\
& = \sum_{w,z \in \bits^n} \Exp_{(a,b),(a',b') \sim Q} (-1)^{\ang{aw,bz}+\ang{a'w,b'z}} \nonumber\\
& = \sum_{w,z \in \bits^n} \Exp_{(a,b),(a',b') \sim Q} (-1)^{\ang{w,a^{-1}bz}-\ang{w,(a')^{-1}b'z}} \nonumber\\
& = \Exp_{c,c' \sim Q'} \sum_{w,z \in \bits^n} (-1)^{\ang{w,(c-c')z}} \nonumber\\
& = \Exp_{c,c' \sim Q'} \sum_{z \in \bits^n} 2^n \cdot \1\paren*{(c-c')z=\bf0} \nonumber\\
& = 2^n \cdot \Exp_{c,c' \sim Q'} \sum_{z \in \bits^n} \1\paren*{c^{-1}c'z=z} \nonumber\\
& = 2^n \cdot \Exp_{c,c' \sim Q'} 2^{\cyc(c^{-1}c')}, \label{eqn:prob-O-g}
\end{align}
where in the last step, we use an argument similar to the one in the proof of \cref{lem:prob-D}.

Plugging \cref{eqn:prob-O-f,eqn:prob-O-g} into \cref{eqn:prob-O}, and by \cref{clm:cyc} for $m=1$, we have that
\begin{align*}
& \Pr(O=\ang{W_r,Z_r} \mid \gamma_{r-1},m_r)-\Pr(O \ne \ang{W_r,Z_r} \mid \gamma_{r-1},m_r) \\
& \hspace{1cm} \le \paren*{2^{-2n+\Phi_{r-1}} \cdot 2^n \cdot 2^{\frac{\bbD_\infty(A_r,B_r \mid \gamma_{r-1},m_r)+\cO(n)}{\log n}}}^{\frac{1}{2}} \\
& \hspace{1cm} \le \paren*{2^{-n+\Phi_{r-1}} \cdot 2^{\frac{L_r+Kn}{\log n}}}^{\frac{1}{2}} \\
& \hspace{1cm} = \paren*{2^{-n+\Phi_r}}^{\frac{1}{2}},
\end{align*}
for any $(\gamma_{r-1},m_r)$ satisfying  both $\oE_{r-1}$ and $\oE'_r$, and sufficiently large constant $K$.
This concludes the proof by averaging over $(\gamma_{r-1},m_r)$,
as \[\Phi_r \le \frac{rL+Krn}{\log n} = o(n),\] for $L=o(n/r \cdot \log n)$ and $r=o(\log n)$.
\end{proof}

Combining \cref{lem:prob-Ep,lem:prob-E,lem:prob-O} yields the desired $\Omega(n/r \cdot \log n)$ lower bound.
Indeed, we have that
\begin{align*}
& \Pr(O=\ang{W_r,Z_r}) \\
& \hspace{1cm} \le \Pr(O=\ang{W_r,Z_r} \mid \oE_{r-1},\oE'_r) + \Pr(\cE'_r) + \Pr(\cE_{r-1}) \\
& \hspace{1cm} \le \frac{1}{2} + 2^{-\Omega(n)} + 2^{-n} + \Pr(\cE_{r-1}) \tag{by \cref{lem:prob-Ep,lem:prob-O}}\\
& \hspace{1cm} \le \frac{1}{2} + 2^{-\Omega(n)} + 2^{-n} + \Pr(\cE_{r-1} \mid \oE_{r-2},\oE'_{r-1}) + \Pr(\cE'_{r-1}) + \Pr(\cE_{r-2}) \\
& \hspace{1cm} \le \frac{1}{2} + 2^{-\Omega(n)} + 2^{-n} + 2^{-\frac{n}{\log n}} + 2^{-n} + \Pr(\cE_{r-2}) \tag{by \cref{lem:prob-Ep,lem:prob-E}}\\
& \hspace{1cm} \le \frac{1}{2} + 2^{-\Omega(n)} + 2^{-n} + 2^{-\frac{n}{\log n}} + 2^{-n} + \Pr(\cE_{r-2} \mid \oE_{r-3},\oE'_{r-2}) + \Pr(\cE'_{r-2}) + \Pr(\cE_{r-3}) \\
& \hspace{2cm} \vdots \\
& \hspace{1cm} \le \frac{1}{2} + 2^{-\Omega(n)} + 2^{-n} + (r-1) \cdot \paren*{2^{-\frac{n}{\log n}} + 2^{-n}} + \Pr(\cE_0) \\
& \hspace{1cm} = \frac{1}{2} + o(1),
\end{align*}
contradicting the error probability $\varepsilon \le 1/3$.

\subsection{Proof of \cref{thm:prob}}
\label{sec:prob-full}

We now consider any protocol with communication $L=n/1000 \cdot \log n$.
The major bottleneck of the analysis in \cref{sec:prob-weak} is that the min-entropy loss in $(A_i,B_i)$ is upper bounded only in terms of the total length $L_i$ of all previous messages, which could be as large as $L$.
Then, a total loss of $rL/\log n$ in $\Phi_r$ would fail the advantage bound in \cref{lem:prob-O}.

Intuitively, however, at most $L_i$ bits of information in total have been revealed across all pairs of permutations $(A_j,B_j)$ for $j \in [i]$.
In other words, $(A_i,B_i)$ should only be charged for the actual amount of revealed information, instead of the maximum possible in the worst case.
After a closer look at the proof in \cref{sec:prob-weak}, we can realize that the full version of \cref{thm:prob} would follow if the upper bound in \cref{lem:prob-D} were in terms of the standard Shannon entropy loss $\bbD_1(Q)$ rather than min-entropy loss $\bbD_\infty(Q)$.
Indeed, although $\bbD_1(A_i,B_i \mid \gamma_{i-1},m_i)$ could be as large as $L$ individually for each $i$, its sum over all $i$ should be at most $L$ in expectation, by chain rule of Shannon entropy (\cref{fact:chain-H}).
This would yield a total loss of only $L/\log n$ in $\Phi_r$.

Unfortunately, the difference between Shannon entropy and min-entropy can be unbounded in general, and may be at least $\Omega(\log(n!))=\Omega(n\log n)$ in our case.
Charging it in every round would immediately destroy all the analysis.

Nevertheless, it turns out we can actually guarantee that Shannon entropy and min-entropy are almost the same, by associating each round with a ``virtual message'' $\wM_i$.
In order to properly define $\wM_i$, we need to change the definition of $\Gamma_i$ at the same time.
Concretely, starting with $\Gamma_0=\varnothing$, recursively define
\[\wM_i=\floor{-\log \Pr(A_i,B_i \mid \gamma_{i-1},m_i)} \qquad \text{and} \qquad \Gamma_i=(X,Y,A_{<i},B_{<i},M'_{\le i}),\]
where $M'_i=(M_i,\wM_i)$ is the new ``combined message''.
Again, we emphasize that $\wM_i$ is only for analysis purposes and is not communicated by the parties throughout the protocol.
At the same time, it is also worth noting that the owner of $(A_i,B_i)$ can actually compute $\wM_i$ based only on its own input and the partial transcript, including virtual messages, so far.
This is because the rectangle property of deterministic communication protocols ensures that the inputs of the two parties are always independent even conditioned on any partial transcript.
Consequently, $\wM_i$ can be treated just like the actual message $M_i$ in the analysis, and in particular preserves the rectangle property itself.
Also, no assumption on fixed message lengths is required.

The following \cref{lem:prob-wM} shows that this virtual message $\wM_i$ has negligible impact on the communication complexity.

\begin{lemma}
\label{lem:prob-wM}
For $i \in [r]$, it holds that \[\log(\wM_i+1)=\cO(\log n).\]
\end{lemma}

\begin{proof}
We prove the lemma in the case of odd $i$.
The case of even $i$ is similar.

Due to the rectangle property of deterministic communication protocols, Alice always holds a uniform distribution on $\supp(\set{(A_{i'},B_{i'})}_{\text{odd } i' \ge i} \mid \gamma_{i-1},m_i)$ after conditioning on any $(\Gamma_{i-1},M_i)=(\gamma_{i-1},m_i)$.
As a result, any possible $(A_i,B_i)$ has marginal probability at least $(n!)^{-2r}$, as the universe is of size at most $(n!)^{2r}$.
This concludes the proof as \[\log\log(n!)^{2r}=\cO(\log(rn\log n))=\cO(\log n).\]
\end{proof}

By \cref{lem:prob-wM}, the extra length caused by $\wM_i$ is at most $\cO(\log(\wM_i+1))=\cO(\log n)$.
More importantly, revealing $\wM_i$ partitions $\supp(A_i,B_i \mid \gamma_{i-1},m_i)$ into different scales, and all probabilities within the same scale differ by at most a factor of $2$.
Specifically, we have that
\begin{align}
\bbD_\infty(A_i,B_i \mid \gamma_{i-1},m'_i) \le \bbD_1(A_i,B_i \mid \gamma_{i-1},m'_i)+1. \label{eqn:prob-D}
\end{align}

The full proof of \cref{thm:prob} essentially follows by replacing $M_i$ with the combined message $M'_i$ in the analysis in \cref{sec:prob-weak}.
Indeed, change the definition of $\Phi_i$ to be \[\Phi_i=\sum_{j \in [i]} \frac{\bbD_1(A_j,B_j \mid \gamma_{j-1},m'_j) + Kn}{\log n},\] which is now a random variable depending on $(\gamma_{i-1},m'_i)$.
The event $\cE_i$ is now defined with respect to the updated threshold $\Phi_i$.
In the meantime, we no longer use the events $\cE'_i$.
Instead, define a single new event $\cE'$ to be \[\sum_{i \in [r]} \bbD_1(A_i,B_i \mid \gamma_{i-1},m'_i) > 20L.\]

With the updated definitions, analogs to \cref{lem:prob-Ep,lem:prob-E,lem:prob-O} can be shown using quite similar arguments.
We do not repeat all the proofs and only sketch a few required adaptations below.
The new statements are summarized in \cref{lem:prob-full}.

\begin{lemma}
\label{lem:prob-full}
It holds that
\begin{enumerate}
    \item\label{item:prob-full-Ep} \[\Pr(\cE') \le \frac{1}{10}.\]
    \item\label{item:prob-full-E} For $i \in [r]$, \[\Pr(\cE_i \mid \oE_{i-1}) \le 2^{-\frac{n}{\log n}}.\]
    \item\label{item:prob-full-O} \[\Pr(O=\ang{W_r,Z_r} \mid \oE_{r-1},\oE') \le \frac{1}{2} + 2^{-\Omega(n)}.\]
\end{enumerate}
\end{lemma}

\begin{proof}
Since the input is uniformly random at the beginning of the protocol, we have that
\begin{align*}
& \Exp_{x,y,a_{\le r},b_{\le r},m'_{\le r}} \sum_{i \in [r]} \bbD_1(A_i,B_i \mid \gamma_{i-1},m'_i) \\
& \hspace{1cm} = \sum_{i \in [r]} \bbI(A_i,B_i;\Gamma_{i-1},M'_i) \\
& \hspace{1cm} = \bbI(A_1,B_1;M'_1) \\
& \hspace{2cm} + \sum_{i \in [2,r]} \paren*{\bbI(A_i,B_i;X,Y,A_{<(i-1)},B_{<(i-1)}) + \bbI(A_i,B_i;M'_{\le i} \mid X,Y,A_{<(i-1)},B_{<(i-1)})} \tag{by chain rule of mutual information (\cref{fact:chain-I})}\\
& \hspace{1cm} = \bbI(A_1,B_1;M'_1) + \sum_{i \in [2,r]} \bbI(A_i,B_i;M'_{\le i} \mid X,Y,A_{<(i-1)},B_{<(i-1)}) \\
& \hspace{1cm} \le \bbI(A_1,B_1;M'_1) + \sum_{i \in [2,r]} \bbI(A_i,B_i;M'_{\le i} \mid X,Y,A_{<i},B_{<i}) \tag{by \cref{fact:I-inc}}\\
& \hspace{1cm} \le \bbI(X,Y,A_1,B_1;M'_{\le r}) + \sum_{i \in [2,r]} \bbI(A_i,B_i;M'_{\le r} \mid X,Y,A_{<i},B_{<i}) \\
& \hspace{1cm} = \bbI(X,Y,A_{\le r},B_{\le r};M'_{\le r}) \tag{by chain rule of mutual information (\cref{fact:chain-I})}\\
& \hspace{1cm} \le \bbH_1(M'_{\le r}) \\
& \hspace{1cm} \le \bbH_1(M_{\le r}) + \bbH_1(\wM_{\le r}) \\
& \hspace{1cm} \le L + \cO(r\log n) \tag{by \cref{lem:prob-wM}}\\
& \hspace{1cm} \le 2L.
\end{align*}
\cref{item:prob-full-Ep} then follows from Markov's inequality.

For \cref{item:prob-full-E}, almost the same proof as for \cref{lem:prob-E} applies, except that the upper bound on $\bbD_2((A_i,B_i)^{(w,z)} \mid \gamma_{i-1},m'_i)$ now comes from \cref{lem:prob-D,eqn:prob-D}, instead of being guaranteed by $\oE'_i$.

\cref{item:prob-full-O} is shown by a similar argument to the one for \cref{lem:prob-O}, where $\oE'$ now ensures \[\Phi_r \le \frac{20L+Krn}{\log n} \le \frac{n}{10},\] for $r=o(\log n)$.
\end{proof}

Finally, we are ready to prove the full version of \cref{thm:prob}.

\begin{proof}[Proof of \cref{thm:prob}]
By \cref{lem:prob-full}, we have that
\begin{align*}
\Pr(O=\ang{W_r,Z_r})
& \le \Pr(O=\ang{W_r,Z_r} \mid \oE_{r-1},\oE') + \Pr(\cE') + \Pr(\cE_{r-1}) \\
& \le \frac{1}{2} + 2^{-\Omega(n)} + \frac{1}{10} + \Pr(\cE_{r-1}) \tag{by \cref{item:prob-full-Ep,item:prob-full-O}}\\
& \le \frac{3}{5} + 2^{-\Omega(n)} + \Pr(\cE_{r-1} \mid \oE_{r-2}) + \Pr(\cE_{r-2}) \\
& \le \frac{3}{5} + 2^{-\Omega(n)} + 2^{-\frac{n}{\log n}} + \Pr(\cE_{r-2}) \tag{by \cref{item:prob-full-E}}\\
& \le \frac{3}{5} + 2^{-\Omega(n)} + 2^{-\frac{n}{\log n}} + \Pr(\cE_{r-2} \mid \oE_{r-3}) + \Pr(\cE_{r-3}) \\
& \hspace{1cm} \vdots \\
& \le \frac{3}{5} + 2^{-\Omega(n)} + (r-1) \cdot 2^{-\frac{n}{\log n}} + \Pr(\cE_0) \\
& = \frac{3}{5} + o(1),
\end{align*}
a contradiction to the assumption $\epsilon \le 1/3$.
This concludes the proof.
\end{proof}

%% file: info.tex
\section{Basic Tools from Information Theory}
\label{sec:info}

We introduce some definitions from information theory that we use in this paper.
Interested readers are referred to the excellent textbook by Cover and Thomas~\cite{CT06} for an introduction to information theory.
For notational convenience, we may identify a random variable with its distribution.

The \emph{Shannon entropy} of a random variable $X$, with distribution $\mu$ supported on universe $\cU$, is defined to be \[\bbH(X) = \sum_{x \in \cU} \mu(x) \cdot \log \frac{1}{\mu(x)}.\]
The \emph{conditional entropy} of $X$ conditioned on $Y$ is defined as \[\bbH(X \mid Y) = \Exp_{y \sim Y} \bbH(X \mid Y=y).\]
For two distributions $\mu,\nu$ supported on the same universe $\cU$, the \emph{Kullback-Leibler (KL) divergence} of $\mu$ from $\nu$ is defined to be \[\kl(\mu \Vert \nu) = \sum_{x \in \cU} \mu(x) \cdot \log\frac{\mu(x)}{\nu(x)}.\]

For $\alpha \in [0,1) \cup (1,\infty)$, the \emph{\renyi{} entropy} of order $\alpha$ of a random variable $X$, with distribution $\mu$ supported on universe $\cU$, is defined to be \[\bbH_\alpha(X) = \frac{1}{1-\alpha} \cdot \log\sum_{x \in \cU} \mu(x)^\alpha.\]
For $\alpha=1$, taking the limit recovers Shannon entropy: \[\bbH_1(X) = \lim_{\alpha \to 1} \bbH_\alpha(X) = \bbH(X).\]
The \emph{min-entropy} of $X$ is defined to be \[\bbH_\infty(X) = \lim_{\alpha \to \infty} \bbH_\alpha(X) = - \log \max_{x \in \cU} \mu(x).\]
The case of $\alpha=2$ is also known as \emph{collision entropy}: \[\bbH_2(X) = - \log \sum_{x \in \cU} \mu(x)^2,\] as the sum is exactly the probability of collision for two independently drawn samples.

For two distributions $\mu,\nu$ supported on the same universe $\cU$, the \emph{\renyi{} divergence} of order $\alpha$ of $\mu$ from $\nu$ is defined to be \[\bbD_\alpha(\mu \Vert \nu) = \frac{1}{\alpha-1} \cdot \log \sum_{x \in \cU} \mu(x)^\alpha \nu(x)^{1-\alpha},\] for $\alpha \in (0,1) \cup (1,\infty)$, and extended to $\alpha=1,\infty$ by limits.
$\bbD_1(\cdot \Vert \cdot)$ is exactly KL divergence.
If $\nu$ is the uniform distribution on $\cU$, we simply write \[\bbD_\alpha(\mu) = \bbD_\alpha(\mu \Vert \nu) = \log\card{\cU} - \bbH_\alpha(\mu).\]

The \emph{mutual information} between two random variables $X,Y$ is defined to be \[\bbI(X;Y)=\bbH(X)-\bbH(X \mid Y)=\bbH(Y)-\bbH(Y \mid X),\]
and the \emph{conditional mutual information} between $X,Y$ conditioned on $Z$ is defined as \[\bbI(X;Y \mid Z)=\Exp_{z \sim Z} \bbI(X;Y \mid Z=z).\]

Throughout, we use the following basic properties.

\begin{fact}[Chain rule of Shannon entropy]
\label{fact:chain-H}
For random variables $A,B,C$, it holds that \[\bbH(A,B \mid C) = \bbH(A \mid B,C) + \bbH(B \mid C).\]
\end{fact}

\begin{fact}[Chain rule of mutual information]
\label{fact:chain-I}
For random variables $A,B,C,D$, it holds that \[\bbI(A,B;C \mid D) = \bbI(A;C \mid B,D) + \bbI(B;C \mid D).\]
\end{fact}

\begin{fact}
\label{fact:I-inc}
For random variables $A,B,C,D$ such that $A \perp D \mid C$, it holds that \[\bbI(A;B \mid C) \le \bbI(A;B \mid C,D).\]
\end{fact}

\begin{fact}
\label{fact:I-dec}
For random variables $A,B,C,D$ such that $A \perp D \mid B,C$, it holds that \[\bbI(A;B \mid C) \ge \bbI(A;B \mid C,D).\]
\end{fact}

%% file: pcp.tex
\section{Round-Efficient Protocols for $\pcp$}
\label{sec:pcp}

We present round-efficient protocols for $\pcp$ and $\bpcp$ in this section.
For simplicity, only the receiver of the last message is required to output for protocols in this section.
It only changes the round complexity by $1$.

\begin{theorem}
\label{thm:pcp}
For $n,r \ge 1$, and $k \le r^2$, there exists an $r$-round deterministic protocol, starting with Bob, for $\pcp_{n,k}$ with $\cO(nr^2\log\log n)$ communication.

Moreover, for $k \le r^2+r$, there exists an $r$-round deterministic protocol, starting with Bob, for $\bpcp_{n,k}$ with $\cO(nr^2\log\log n)$ communication.
\end{theorem}

The protocol for proving \cref{thm:pcp} utilizes \cref{alg:pcp-delay} as a subroutine, which solves $\pcp_{n,2\tau-1}$ in $\tau$ rounds with $\cO(n\tau\log\log n)$ communication, assuming the proper party starts communication.
More crucially, it has the ``delayed starting vertex'' property, as formalized in \cref{lem:pcp-delay}.

\begin{algorithm}[tp]
    \caption{A $\tau$-round protocol for $\pcp_{n,2\tau-1}$.}
    \label{alg:pcp-delay}
    \begin{algorithmic}[1]
        \State \label{line:pcp-delay-init-p}Set $p \gets 1$ for Alice and $p \gets 0$ for Bob.
        \State Initialize $L=\ceil{\log n}$, and $\chi_i(v) \gets 0$ for $i \in [\tau]$ and $v \in [n]$.
        \If{$p \equiv \tau \pmod{2}$} \label{line:pcp-delay-init-if}
            \State \label{line:pcp-delay-init-t}Set $\chi_\tau(v) \gets v \bmod L$ for $v \in [n]$.
            \State \label{line:pcp-delay-init-s}Set $\chi_{\tau-1}(v) \gets \chi_\tau(f_\tau(v))$ for $v \in [n]$.
        \EndIf
        \For{$i=1,\dots,\tau-1$}
            \If{$p \equiv \tau-i+1 \pmod{2}$} \label{line:pcp-delay-i-if}
                \State \label{line:pcp-delay-i-send}Send $\chi_{\tau-i},\chi_{\tau+i-1}$.
            \Else
                \State Receive $\chi_{\tau-i},\chi_{\tau+i-1}$.
                \State \label{line:pcp-delay-i-s}Set $\chi_{\tau-i-1}(v) \gets \chi_{\tau-i}(f_{\tau-i}(v))$ for $v \in [n]$.
                \State \label{line:pcp-delay-i-t}Set $\chi_{\tau+i}(v) \gets \chi_{\tau+i-1}(f_{\tau+i}^{-1}(v))$ for $v \in [n]$.
            \EndIf
        \EndFor
        \If{$p=1$} \label{line:pcp-delay-r-if}
            \State \label{line:pcp-delay-v0}Get input $v_0$.
            \State \label{line:pcp-delay-r-send}Send $v_0$, and $\chi_0,\chi_{2\tau-1}$.
            \State Send $f_i(v)$ for odd $i \in [2\tau-1]$ and all $v \in [n]$ with $\chi_{i-1}(v)=\chi_0(v_0)$.
        \Else
            \State Receive $v_0$, and $\chi_0,\chi_{2\tau-1}$.
            \State Receive $f_i(v)$ for odd $i \in [2\tau-1]$ and all $v \in [n]$ with $\chi_{i-1}(v)=\chi_0(v_0)$.
            \State \label{line:pcp-delay-out}Output $f_{2\tau-1} \circ \dots \circ f_1(v_0)$.
        \EndIf
        \alglinenoNew{algcommon}
        \alglinenoPush{algcommon}
    \end{algorithmic}
\end{algorithm}

\begin{lemma}
\label{lem:pcp-delay}
For $n,\tau \ge 1$, there exists a $\tau$-round deterministic protocol for $\pcp_{n,2\tau-1}$ with $\cO(n\tau\log\log n)$ communication such that\footnote{In fact, \cref{alg:pcp-delay} can also solve $\pcp_{n,2\tau}$ as Bob owns $f_{2\tau}$ and thus can output correctly in \cref{line:pcp-delay-out}. Nevertheless, due to parity constraints, this extra layer does not help improve \cref{alg:pcp}.}
\begin{enumerate}
    \item \label{item:pcp-delay-1} The first message is sent by Alice for odd $\tau$ and Bob for even $\tau$.
    \item \label{item:pcp-delay-r} The last message is received by Bob.
    \item \label{item:pcp-delay-0} The starting vertex $v_0$ is given to Alice only at the beginning of the last round.
\end{enumerate}
\end{lemma}

\begin{proof}
Fix sufficiently large $n$, and $\tau \ge 1$.
We prove the lemma using \cref{alg:pcp-delay}.
\cref{item:pcp-delay-1,item:pcp-delay-r} follow directly from the initialization in \cref{line:pcp-delay-init-p} and the conditions in \cref{line:pcp-delay-i-if,line:pcp-delay-r-if}.
They also ensure an alternating speaking order for Alice and Bob.
\cref{line:pcp-delay-v0} implies \cref{item:pcp-delay-0}.

The correctness of the protocol is shown by induction on round $i$.
In the base case of $i=0$, the party with $p \equiv \tau \pmod{2}$ holds $f_\tau$ as input.
Due to the condition in \cref{line:pcp-delay-init-if}, it (almost) equally partitions each of layers $\tau$ and $\tau-1$ into $L=\Theta(\log n)$ color classes in \cref{line:pcp-delay-init-t,line:pcp-delay-init-s}.
Furthermore, the color of each vertex $v$ in layer $\tau-1$ is the same as the color of $f_\tau(v)$ in layer $\tau$.
In other words, the coloring is consistent with the pointers for layer $\tau$.

In round $i \in [\tau-1]$, by induction, the sender knows the colorings of layers $\tau-i,\dots,\tau+i-1$, while the receiver knows the same colorings except the two outer layers $\tau-i$ and $\tau+i-1$.
Once the receiver gets the colorings of layers $\tau-i,\tau+i-1$ from the sender, it further extends them to layers $\tau-i-1,\tau+i$ in \cref{line:pcp-delay-i-s,line:pcp-delay-i-t}.
This is well-defined because it holds $f_{\tau-i},f_{\tau+i}$ as input.

In the last round, Alice sends the newly obtained starting vertex $v_0$, along with the colorings of layers $0,2\tau-1$.
Upon receiving the message, Bob now knows the colorings of all layers.
More importantly, for every $i \in [2\tau-1]$, the $i$-th target vertex $v_i$ must have the same color as $v_0$ because the colorings are consistent with all pointers.
Consequently, it suffices for Alice to additionally send her pointers only for vertices with color $\chi_0(v_0)$.
This ensures that Bob has enough knowledge to output correctly.

As to the total amount of communication, only $L=\Theta(\log n)$ colors are used, so sending each coloring costs at most $\cO(n\log L) = \cO(n\log\log n)$ bits in \cref{line:pcp-delay-i-send,line:pcp-delay-r-send}.
In the meantime, each layer has $\Theta(n/L)$ vertices of any specific color.
Sending pointers for these vertices costs $\cO(n/L \cdot \log n) = \cO(n)$ per layer owned by Alice.
In total, the protocol communicates $\cO(n\tau\log\log n)$ bits.
This concludes the proof.
\end{proof}

The ``delayed starting vertex'' property in \cref{item:pcp-delay-0} of \cref{lem:pcp-delay} enables us to run $r$ parallel instances of \cref{alg:pcp-delay} on disjoint layers and gradually chase through more and more layers at once.
Concretely, \cref{alg:pcp} utilizes this intuition to quadratically improve the round complexity.

\begin{algorithm}[tp]
    \caption{An $r$-round protocol for $\pcp_{n,r^2}$.}
    \label{alg:pcp}
    \begin{algorithmic}[1]
        \alglinenoPop{algcommon}
        \State Set $p \gets 1$ for Alice and $p \gets 0$ for Bob.
    	\For{$i=2,\dots,r$ in parallel}
            \State Run in parallel an instance of \cref{alg:pcp-delay} with $\tau=i$ on $f_{(i-1)^2},\dots,f_{i^2-1}$,
            \Statex \hspace{1cm} swapping Alice and Bob for odd $i$.
            \State \label{line:pcp-i-in}Give $v_{(i-1)^2-1}$ to Alice (or Bob for odd $i$) at the beginning of round $i$.
            \State \label{line:pcp-i-out}Get $v_{i^2-1}$ from Bob (or Alice for odd $i$) at the end of round $i$.
        \EndFor
        \If{$p \equiv r^2 \pmod{2}$}
            \State \label{line:pcp-out}Output $f_{r^2}(v_{r^2-1})$.
        \EndIf
        \alglinenoPush{algcommon}
    \end{algorithmic}
\end{algorithm}

\begin{algorithm}[tp]
    \caption{An $r$-round protocol for $\bpcp_{n,r^2+r}$.}
    \label{alg:bpcp}
    \begin{algorithmic}[1]
        \alglinenoPop{algcommon}
        \State Set $p \gets 1$ for Alice and $p \gets 0$ for Bob.
        \If{$p \equiv r^2+r \pmod{2}$}
            \State Initialize $b_{r^2+r-1}(v) \gets f_{r^2+r}(v) \bmod{2}$ for $v \in [n]$.
        \EndIf
        \State Run in parallel an instance of \cref{alg:pcp} on layers $1,\dots,r^2$.
        \For{$i=1,\dots,r-1$}
            \If{$p \equiv r^2+r-i+1 \pmod{2}$} \label{line:bpcp-i-if}
                \State \label{line:bpcp-i-send}Send $b_{r^2+r-i}$.
            \Else
                \State Receive $b_{r^2+r-i}$.
                \State Set $b_{r^2+r-i-1}(v) \gets b_{r^2+r-i}(f_{r^2+r-i}(v))$ for $v \in [n]$.
            \EndIf
        \EndFor
        \If{$p \equiv r^2+1 \pmod{2}$}
            \State \label{line:bpcp-r-send}Send $b_{r^2}$.
        \Else
            \State Receive $b_{r^2}$.
            \State Get $v_{r^2}$ from \cref{alg:pcp}.
            \State Output $b_{r^2}(v_{r^2})$.
        \EndIf
        \alglinenoPush{algcommon}
    \end{algorithmic}
\end{algorithm}

Now, we are ready to prove \cref{thm:pcp}.

\begin{proof}[Proof of \cref{thm:pcp}]
Fix sufficiently large $n$, and $r \ge 1$.
Assume without loss of generality that $k=r^2$.
Since the case of $r=1$ holds trivially, assume $r \ge 2$ as well.
We first prove by induction on $i$ that \cref{alg:pcp} solves $\pcp_{n,k}$.
By \cref{item:pcp-delay-1} of \cref{lem:pcp-delay}, \cref{alg:pcp-delay} starts with Alice for odd $\tau$ and Bob for even $\tau$.
Thus, swapping the roles of Alice and Bob for odd $i$ ensures that all instances of \cref{alg:pcp-delay} start with Bob.

For the base case of $i=1$, Alice knows $v_0$.
For $i \in [2,r]$, layer $(i-1)^2$ is held by Alice for even $i$ and Bob for odd $i$.
Also verify that $i^2-1-(i-1)^2+1=2i-1=2\tau-1$.
Therefore, the $i$-th instance of \cref{alg:pcp-delay} can be run on layers $(i-1)^2,\dots,i^2-1$, with the roles of Alice and Bob swapped for odd $i$.
In round $i$, the $(i-1)$-th instance has already completed, so $v_{(i-1)^2-1}$ is held by Alice for even $i$ and Bob for odd $i$ by inductive hypothesis (or $v_0$ is already known to Alice for $i=2$).
It then implies that \cref{line:pcp-i-in,line:pcp-i-out} are well-defined, due to \cref{item:pcp-delay-0,item:pcp-delay-r} of \cref{lem:pcp-delay}.
After all $r$ rounds, both $v_{r^2-1}$ and $f_{r^2}$ are known to Alice for odd $r$ and Bob for even $r$, so \cref{line:pcp-out} outputs correctly.
The total communication is \[\sum_{\tau=1}^r \cO(n\tau\log\log n) = \cO(nr^2\log\log n).\]

For the bit version, simply propagate the bits backwards from the last layer, in parallel to \cref{alg:pcp}.
This helps to chase through $r$ extra layers.
The algorithm is formalized in \cref{alg:bpcp}.
Observe that $r^2+r=r(r+1)$ is always even, so the condition in \cref{line:bpcp-i-if} ensures that backward propagation starts with Bob as well.
The additional communication is only $\cO(n)$ per layer in \cref{line:bpcp-i-send,line:bpcp-r-send}.
As to correctness, the receiver of round $i$ owns $f_{r^2+r-i}$ and thus can derive $b_{r^2+r-i-1}$ by induction.
After $r-1$ rounds of backward propagation, $b_{r^2}$ is known to Alice for even $r$ and Bob for odd $r$.
Round $r$ communicates it to the other party and thus ensures correct output, concluding the proof.
\end{proof}

%% file: prob-sublinear.tex
\section{Lower Bound for $\prob$ with More Rounds}
\label{sec:prob-sublinear}

In this section, we prove the following strengthening of \cref{thm:prob}.

\begin{theorem}
\label{thm:prob-sublinear}
For $n \ge 1$, $\eta \in (0,1)$, and $r=\cO(n^\eta)$, any $r$-round protocol for $\prob_{n,r}$ with constant error probability less than $1/2$ requires at least $\bfR^{(r)}(\prob_{n,r}) = \Omega_\eta(n \log n)$ communication.
\end{theorem}

The only constraint of the analysis in \cref{sec:prob} that limits $r=o(\log n)$ is the $\cO(rn/\log n)$ loss in $\Phi_r$, which can be traced back to \cref{lem:prob-D}.
Therefore, the key to proving \cref{thm:prob-sublinear} is a stronger alternative to \cref{lem:prob-D}, as formalized in \cref{lem:prob-D-sublinear}.

\begin{lemma}
\label{lem:prob-D-sublinear}
For $\zeta \in (0,1)$ and any distribution $Q$ on $\cP_n \times \cP_n$, it holds that \[\Exp_{w,z \sim \bits^n} 2^{\bbD_2(Q^{(w,z)})} \le 2^{\frac{\bbD_\infty(Q)}{\zeta\log n}+\cO_\zeta(n^\zeta)},\]
where $Q^{(w,z)}$ is the distribution of $(aw,bz) \in \bits^{2n}$ when $(a,b) \sim Q$.
\end{lemma}

The proof of \cref{lem:prob-D-sublinear} uses \holder{}'s inequality, which can be viewed as a generalization of Cauchy-Schwarz inequality.

\begin{lemma}[\holder{}'s inequality]
\label{lem:holder}
For $p,q > 1$ such that $1/p+1/q=1$, and any two real-valued random variables $X,Y$, it holds that \[\Exp\abs{XY} \le \paren*{\Exp\abs{X}^p}^{\frac{1}{p}}\paren*{\Exp\abs{Y}^q}^{\frac{1}{q}}.\]
\end{lemma}

\begin{proof}[Proof of \cref{lem:prob-D-sublinear}]
Arguments similar to the proofs of \cref{lem:prob-D,clm:cyc} show that it suffices to upper bound \[\Exp_{(a,b) \sim Q} 2^{\cyc(a)+\cyc(b)}.\]
To this end, for $p,q > 1$ such that $1/p+1/q=1$, we have that
\begin{align}
\Exp_{(a,b) \sim Q} 2^{\cyc(a)+\cyc(b)}
& = \Exp_{a,b \sim \cP_n} (n!)^2 \cdot Q(a,b) \cdot 2^{\cyc(a)+\cyc(b)} \nonumber\\
& \le \paren*{\Exp_{a,b \sim \cP_n} (n!)^{2p} \cdot Q(a,b)^p}^{\frac{1}{p}} \paren*{\Exp_{a,b \sim \cP_n} 2^{q \cdot (\cyc(a)+\cyc(b))}}^{\frac{1}{q}}, \label{eqn:prob-holder}
\end{align}
by \holder{}'s inequality (\cref{lem:holder}).
Also observe that \[\Exp_{a,b \sim \cP_n} (n!)^2 \cdot Q(a,b) = (n!)^2 \cdot (n!)^{-2} = 1.\]
Therefore, we can get
\begin{align}
\Exp_{a,b \sim \cP_n} (n!)^{2p} \cdot Q(a,b)^p
& \le \paren*{\max_{a,b \in \cP_n} (n!)^2 \cdot Q(a,b)}^{p-1} \nonumber\\
& \le \paren*{2^{2\log(n!)} \cdot 2^{-\bbH_\infty(Q)}}^{p-1} \nonumber\\
& = 2^{(p-1) \cdot \bbD_\infty(Q)}. \label{eqn:prob-holder-p}
\end{align}
In the meantime, using the cycle decomposition representation of permutations, we also have
\begin{align}
\Exp_{a,b \sim \cP_n} 2^{q \cdot (\cyc(a)+\cyc(b))}
& = \paren*{\Exp_{a \sim \cP_n} 2^{q \cdot \cyc(a)}}^2 \nonumber\\
& = \prod_{i \in [n]} \paren*{\frac{1}{i} \cdot 2^q + \paren*{1-\frac{1}{i}} \cdot 1}^2 \nonumber\\
& = \prod_{i \in [n]} \paren*{1 + \frac{2^q-1}{i}}^2 \nonumber\\
& \le \exp\paren*{2 \cdot \sum_{i \in [n]} \frac{2^q-1}{i}} \nonumber\\
& = \exp\paren*{\cO(2^q \cdot \ln n)}. \label{eqn:prob-holder-q}
\end{align}
To conclude the proof, simply plug \cref{eqn:prob-holder-p,eqn:prob-holder-q}, as well as $(p-1)/p=1/q$ and $q=\zeta\log n$, into \cref{eqn:prob-holder}.
\end{proof}

Now, we are ready to prove \cref{thm:prob-sublinear}.

\begin{proof}[Proof of \cref{thm:prob-sublinear}]
Fix $\zeta=(1-\eta)/2 > 0$, so that $rn^{\zeta}=o(n)$, and sufficiently large $n$.
Change the definition of $\Phi_i$ to be \[\Phi_i = \sum_{j \in [i]} \paren*{\frac{\bbD_1(A_j,B_j \mid \gamma_{j-1},m'_j)}{\zeta\log n}  + Kn^\zeta},\] for some sufficiently large constant $K$, and adjust the event $\cE_i$ to be defined with respect to the updated threshold $\Phi_i$.
Then, conditioned on $\oE'$, we can get \[\Phi_r \le \frac{20L}{\zeta\log n} + Krn^\zeta \le \frac{n}{10},\] for $L=\zeta n/1000 \cdot \log n$.
Therefore, it can be verified that \cref{lem:prob-full} still holds, except that $2^{-n/\log n}$ becomes $2^{-n^\zeta/\log n}$ in \cref{item:prob-full-E}.
The same argument then concludes the proof.
\end{proof}